\documentclass[12pt]{article}
\usepackage[english]{babel}

\usepackage{mathtools}
\mathtoolsset{showonlyrefs}
\numberwithin{equation}{section}
\usepackage{xcolor}
\usepackage{graphicx}
\usepackage{listings}
\usepackage{algorithm}
\usepackage{algpseudocode}

\newcommand{\bm}{}
\usepackage[colorlinks=true, allcolors=blue]{hyperref}
\usepackage[table]{xcolor} 
\usepackage{colortbl,xcolor,booktabs,diagbox,siunitx}
\usepackage{booktabs}
\usepackage{siunitx}
\usepackage{caption}
\usepackage{array}
\usepackage{multirow}
\usepackage{diagbox}
\usepackage{enumitem}

\definecolor{rowe2}{gray}{0.96}       
\definecolor{colN}{RGB}{235,246,252}  
\definecolor{cross}{RGB}{225,235,240} 

\newcommand{\rowe}{\rowcolor{rowe2}}        
\newcommand{\colmark}[1]{\cellcolor{colN}#1} 
\newcommand{\crosscell}[1]{\cellcolor{cross}#1} 

\RequirePackage{amsthm,amsmath,amsfonts,amssymb}
\RequirePackage[authoryear]{natbib}

\RequirePackage{subcaption}

\DeclareMathOperator{\Log}{log}
\DeclareMathOperator{\Exp}{exp}

\usepackage{scalerel}
\newcommand{\E}{{\mathbb E}}
\newcommand{\R}{\mathrm{I\!R}}

\newcommand{\iw}{\mathcal{IW}}
\newcommand{\siw}{\mathcal{SIW}}

\newcommand{\ml}{{M_l}}
\newcommand{\nl}{{N_l}}
\newcommand{\mtl}{{M_{T,l}}}
\newcommand{\ms}{{\pi}}

\newcommand{\bone}{\mathbf{1}}

\newcommand{\cov}{covariance }

\newcommand{\modif}{\textcolor{blue}}

 \usepackage{amsmath}

\usepackage{fancyhdr}
\usepackage{graphicx}
\usepackage{color}
\newtheorem{theorem}{Theorem}[section]
\theoremstyle{definition}
\newtheorem{definition}{Definition}[section]
\newtheorem{lemma}[theorem]{Lemma}

\newtheorem{prop}{Proposition}[section]
\newtheorem{exmp}{Example}[section]

\usepackage[left=2.5cm,right=2.5cm,top=2.5cm,bottom=2.5cm]{geometry}
\begin{document}
\renewcommand{\headrulewidth}{0pt}

\fancyhead[R]{ }


\begin{center}
{\Large
	{\sc  New sampling approaches for Shrinkage Inverse-Wishart distribution}
}
\bigskip

 Yiye Jiang
\bigskip

{\it
Univ. Grenoble Alpes, Inria, CNRS, Grenoble INP, LJK, 38000 Grenoble, France, \\
yiye.jiang@inria.fr

}
\end{center}
\bigskip


{\bf Abstract.} 
In this paper, we propose new sampling approaches for the \textit{Shrinkage Inverse-Wishart} ($\siw$) distribution \citep{berger2020Bayesian}, a generalized family of the Inverse-Wishart distribution. It offers a flexible prior for covariance matrices and remains conjugate to the Gaussian likelihood, similar to the classical Inverse-Wishart. 
Despite these advantages, sampling from $\siw$ remains challenging. The existing algorithm relies on a nested Gibbs sampler, which is slow and lacks rigorous theoretical analysis of its convergence. We propose a new algorithm based on the Sampling Importance Resampling (SIR) method, which is significantly faster and comes with theoretical guarantees on convergence rates.
A known issue with SIR methods is the large discrepancy in importance weights, which occurs when the proposal distribution has thinner tails than the target. In the case of $\siw$, certain parameter settings can lead to such discrepancies, reducing the robustness of the output samples. To sample from such $\siw$ distributions, we robustify the proposed algorithm by including a clipping step to the SIR framework which transforms large importance weights. We provide theoretical results on the convergence behavior in terms of the clipping size, and discuss strategies for choosing this parameter via simulation studies. The robustified version retains the computational efficiency of the original algorithm.


\bigskip\bigskip

\section{Introduction}
A well-known criticism of Inverse-Wishart ($\iw$) distribution is that it concentrates little mass over covariance matrices with small eigengaps. To rebalance the mass of the distribution, \cite{berger2020Bayesian} introduces Shrinkage Inverse-Wishart ($\siw$) distribution, which is a larger family of probability distributions. The definition is recalled as follows.
\begin{definition}\label{def: siw}
Let $S^K_{++}$ denote $\{\Sigma \in \R^{K \times K}: \Sigma \mbox{ is positive definite}\}$, and $(S^K_{++}, \mathcal{B})$ the measurable space with $\mathcal{B}$ being the Borel $\sigma$-algebra induced by the usual sub-topology of $S^K_{++}$. The Shrinkage Inverse-Wishart ($\siw$) distribution, denoted by $\pi$, is defined as a probability measure on $(S^K_{++}, \mathcal{B})$, which admits a probability density function $\pi(\bm \Sigma | \Psi, \nu, b)$ such that
\begin{equation}\label{eq: siw den}
    \pi(\bm \Sigma | \Psi, \nu, b) \propto \frac{\exp\left(\mbox{tr}\left( -\frac{1}{2}\bm \Sigma^{-1}\Psi\right)\right)}{|\bm \Sigma|^\nu \prod_{i < j} (\lambda_i - \lambda_j)^b}\bone_{\{\lambda_1 > \ldots > \lambda_K > 0\}},
\end{equation} 
where $\lambda_1, \cdots, \lambda_K$ are the eigenvalues of $\bm \Sigma \in S^K_{++}$, $\Psi \in S^K_{++}$, $b \in [0,1]$ and the degree of freedom $\nu$ is real.
\end{definition}
Equation \eqref{eq: siw den} displays the $\iw$ density's kernel, namely an unnormalized probability density function, divided by $\prod_{i < j} (\lambda_i - \lambda_j)^b, \; b \in [0, 1]$, with $b$ additional parameter controlling the mass on \cov matrices of small eigengaps.  

In this paper, we propose a new sampling approach for $\siw$ with $b=1$. The new method is based on sampling importance resampling (SIR) and is without Markov chain Monte Carlo. By contrast, the algorithm proposed in the original paper \cite{berger2020Bayesian} is based on a nested Gibbs sampling procedure. At each iteration of the outer Gibbs procedure, an inner procedure needs to be performed in order to produce a new sample. The proposed algorithm does not require such nested procedure. The sampling and resampling steps perform two independent direct samplings, with the one at the resampling step simply from a Multinomial distribution. Therefore, the proposed sampling is superior in terms of computation time. In addition, we provide theoretical guarantees with convergence rates which are not available in \cite{berger2020Bayesian}. 

A known problem with sampling and importance resampling is large discrepancy in importance weights, which happens when the proposal distribution has a thinner tail than the target distribution, and it will make the estimators obtained from the corresponding samples have high variance. In our case, some $\siw$ parameter values can result in the large discrepancy. To address the problem, we consider an additional clipping step in our approach to transform values of large weights.  Accordingly, we derive the convergence rates in terms of the clipping size. We discuss the strategies on choosing clipping sizes with simulations. The robustified algorithm has the same advantage in running time as the original proposed algorithm.

Finally, we developed an exact sampling solution for a sub-family of $\siw$ where $\Psi = cI_K, \; c > 0$. We tested the exactitude of the algorithm by comparing the sample moments with theoretical values provided in \cite{berger2020Bayesian}. In terms of running time, the proposed algorithm is at least $100$ faster than \cite{berger2020Bayesian}, with more advantage for higher $K$, which thus allows to go to large dimension as $K=1000$.

\paragraph{Organization of the paper.} In Section \ref{sec: sampling}, we develop the three sampling algorithms: in Section \ref{sec: algo Psi simplified}, we first derive the exact sampling for the simplified case where $\Psi = cI_K, \, c > 0$, in Section \ref{sec: algo Psi general}, we derive the algorithm for general $\Psi$, in Section \ref{sec: algo Psi general robust}, we first discuss the cases when the importance weights have a large discrepancy, and then derive a robust version of the proposed algorithm by adopting weight clipping. Finally in Section \ref{sec: num}, we perform experiments to validate the theoretical results on the convergence rates, to evaluate the running time, and to investigate the impact of the clipping size on the performance of the robust algorithm, hence suggesting strategies for clipping size tuning.

\section{Methods}\label{sec: sampling}


Let us use the eigen-representation \citep[Equation (2.5)]{berger2020Bayesian} of density's kernel \eqref{eq: siw den}.  Consider the eigen-decomposition $\Sigma = \Gamma \Delta \Gamma^\top$, where $\Delta = \mbox{diag}(\lambda_1, \ldots, \lambda_K)$, and $\mathcal{O}_K$ is the set of orthonormal matrices in $\R^{K \times K}$, the eigen-representation writes as:
\begin{equation}\label{eq: siw den eigen rep}
    \pi(\Delta, \Gamma| \nu, \Psi, b) \propto \frac{\exp\left(\mbox{tr}\left( -\frac{1}{2}\Gamma \Delta^{-1} \Gamma^\top\Psi\right)\right)}{|\Delta|^\nu \prod_{i < j} (\lambda_i - \lambda_j)^{b-1}} \bone_{\{\lambda_1 > \ldots > \lambda_K > 0\}}\bone_{\{\Gamma \; \in \; \mathcal{O}_K\}},
\end{equation} 
\modif{The decrease in the power from Equations \eqref{eq: siw den} to \eqref{eq: siw den eigen rep} comes from the Jacobian of eigen-decomposition:}
\begin{equation}
    \frac{\partial \Sigma}{\partial (\Delta, \Gamma)} = \prod_{i < j} (\lambda_i - \lambda_j).
\end{equation}
When $b=1$, Equation \eqref{eq: siw den eigen rep} is simplified to 
\begin{equation}\label{eq: siw b=1}
    \pi(\Delta, \Gamma| \nu, \Psi, 1) \propto \frac{\exp\left(\mbox{tr}\left( -\frac{1}{2}\Gamma \Delta^{-1} \Gamma^\top\Psi\right)\right)}{|\Delta|^\nu } \bone_{\{\lambda_1 > \ldots > \lambda_K > 0\}}\bm1_{\{\Gamma \; \in \; \mathcal{O}_K\}}.
\end{equation} 
Fixing $\Gamma$ and ignoring the monotonicity constraint on $\lambda_i$, the kernel corresponds to the product of $K$ independent Inverse-Gamma distributions on $\lambda_i$. This motivates us to write the joint distribution of $(\Delta, \Gamma)$ in marginal distribution of $\Gamma$ multiplied by the conditional distribution of $\Delta$ given $\Gamma$ as follows:
\begin{equation}\label{eq: margi cond}
\begin{aligned}
    \pi(\Delta, \Gamma| \nu, \Psi, 1) &=   \pi(\Delta| \Gamma, \nu, \Psi, 1) \pi(\Gamma| \nu, \Psi, 1) \\
    &\propto \prod_{i = 1}^K k_{\mathcal{IG}}(\lambda_i |\nu -1, \frac{1}{2}\Gamma_{i}^\top \Psi \Gamma_{i}) \bone_{\{\lambda_1 > \ldots > \lambda_K > 0\}} k_{\mathcal{U}\{ \mathcal{O}_K\}}(\Gamma),   
\end{aligned}
\end{equation} 
where $k_{\mathcal{IG}}$ and $k_{\mathcal{U}\{ \mathcal{O}_K\}}$ denote respectively kernels of inverse-gamma distribution and uniform distribution over $\mathcal{O}_K$, and $\Gamma_{i}$ denotes the $i$-th column of $\Gamma$. To make sure the inverse-gammas is well-defined, which is the foundation of the following development, we assume $\nu > 1$.  
Therefore, sampling $(\Delta, \Gamma) \sim \pi(\Delta, \Gamma| \nu, \Psi, 1)$ can be done in two steps:
\begin{enumerate}
\label{eq: target} 
    \item$\bm\Gamma \sim k_{\mathcal{U}\{ \mathcal{O}_K\}}(\Gamma) $ 

    \item 
        $\bm\Delta | \bm\Gamma \sim \prod_{i = 1}^K k_{\mathcal{IG}}(\lambda_i |\nu -1, \frac{1}{2}\bm\Gamma_{i}^\top \Psi \bm\Gamma_{i}) \bone_{\{\lambda_1 > \ldots > \lambda_K > 0\}}.$
\end{enumerate}
In the first step, the uniform distribution is identified in \cite{tropp2012comparison} by the Haar probability measure on Stiefel manifold, where its sampling is solved by the QR decomposition of a random matrix whose entries are sampled independently from $\mathcal{N}(0,1)$. The sampling in the second step is not straightforward except the special case where $\Psi = cI_K,$ with $c > 0$. Therefore, we first derive the sampling solution for this special case in Section \ref{sec: algo Psi simplified}. Then we turn to the general case in Section \ref{sec: algo Psi general}, where the framework of SIR is adopted. We extend it furthermore to a robust version in Section \ref{sec: algo Psi general robust}.

\subsection{Simplified case: $\Psi = cI_{K}, \, c > 0$}\label{sec: algo Psi simplified}

When $\Psi = cI_{K}, \, c > 0$, $\bm \Delta$ and $\bm \Gamma$ are independent, with the marginal distribution of $\bm\Delta$: 
\begin{equation}\label{eq: delta cI}
    \bm \Delta \sim \prod_{i = 1}^K \pi_{\mathcal{IG}}\left(\lambda_i |\nu -1, \frac{c}{2}\right) \bone_{\{\lambda_1 > \ldots > \lambda_K > 0\}}. 
\end{equation}
Given the results on order statistics  recalled in Appendix \ref{thm: order stat}, sampling $\bm\Delta$ from distribution \eqref{eq: delta cI} can be done in two steps:
\begin{enumerate}
    \item $\bm y_i \stackrel{iid}{\sim} \pi_{\mathcal{IG}}\left(\lambda_i |\nu -1, \frac{c}{2}\right), i = 1, \ldots, K $
    \item $\bm \lambda_i := \bm y_{(K-i+1)},$
\end{enumerate}
where $\bm y_{(K)} > \ldots > \bm y_{(1)}$ are the order statistics of $y_i, i = 1, ..., K$. 

Thus, we can obtain an exact sampling from $\siw(\nu, cI_{K}, 1),$ which is summarized in Algorithm \ref{alg: siw Psi=cI}. 
\begin{algorithm}
\caption{Sampling from $\siw(\nu, cI_{K}, 1), \, c > 0$}\label{alg: siw Psi=cI}
\begin{algorithmic}[1]
\Require degree of freedom $\nu>1$, scale $c > 0$, sample size $N$
\Ensure Samples $\{\Sigma^{(n)}\}_{n=1}^N \subset S^K_{++}$
\For{$n = 1,\dots,N$}
  \State Draw $\alpha_{i,j} \stackrel{iid}{\sim} \mathcal{N}(0, 1), \; i,j =1, \ldots, K,$ and form the matrix $A := [\alpha_{i,j}]$
  \State Compute QR decomposition of $A$, denoted by $A = QR$, and set $\Gamma = Q$
  \For{$i = 1,\dots,K$}
    \State Draw $\lambda_i \sim \mathcal{IG}\bigl(\nu-1,\frac{c}{2}\bigr)$
  \EndFor
  \State Sort $(\lambda_i)$ descending to get $(\lambda_{(1)},\dots,\lambda_{(K)})$
  \State Set $\Sigma^{(n)} \leftarrow \Gamma \,\mathrm{diag}(\lambda_{(1)},\dots,\lambda_{(K)})\,\Gamma^\top$
\EndFor

\Return $\{\Sigma^{(n)}\}_{n=1}^N$
\end{algorithmic}
\end{algorithm}

\subsection{General matrix}\label{sec: algo Psi general}
When $\Psi \neq cI_K$, all $\bm\lambda_i$ do not follow the same distribution given $\bm\Gamma$, thus the result of order statistics do not apply. We consider  SIR \citep[Section 7]{cappe2005inference}, which consists in two steps: sampling from the proposal distribution, and resampling from the previous proposal samples with probabilities proportional to the difference between their target and proposal densities. Thus when the target sampling is difficult, one can rely on a proposal distribution which is easier to sample. The SIR framework applied to our sampling problem is detailed as follows. 
\begin{enumerate}
    \item \textbf{Sampling from the proposal:} $M$ independent candidate samples, $\{(\Delta^{(m)}, \Gamma^{(m)})\}_{m=1}^M$, are drawn from a proposal distribution $\tau(\Delta^{(m)}, \Gamma^{(m)})$. For each $(\Delta^{(m)}, \Gamma^{(m)})$, an unnormalized importance weight $w(\Delta^{(m)}, \Gamma^{(m)})$ is defined: 
    $$
w(\Delta^{(m)}, \Gamma^{(m)}) \propto \frac{\pi(\Delta^{(m)}, \Gamma^{(m)})}{\tau(\Delta^{(m)}, \Gamma^{(m)})},
    $$
    where $\pi$ is the $\siw$ density given in Equation \eqref{eq: siw b=1}. Here for both target $\pi$ and proposal $\tau$, we omit the parameters $\Psi$ and $\nu$ in the index for the simplicity of notations. Form: $\Sigma^{(m)} = \Gamma^{(m)}\Delta^{(m)}(\Gamma^{(m)})^T$.

    \item \textbf{Resampling:} $N$ samples are drawn with replacement from the set of $M$ candidates, where the probability of selecting any given candidate $\Sigma^{(m)}$, equivalently  $(\Delta^{(m)}, \Gamma^{(m)})$, is $$\frac{w(\Delta^{(m)}, \Gamma^{(m)})}{\sum_{m=1}^M w(\Delta^{(m)}, \Gamma^{(m)})}.$$
\end{enumerate}
The final output is the $N$ resampled samples. The method consists only in two direct samplings but does not generate exact samples from the target as in the previous section. Nevertheless, when $M$ increases, the final output is closer to the samples drawn from the target. This can be seen from the following calculation. 
Let $\Sigma^*$ be a resampled sample. For any measurable set $A \subset S^K_{++}$, the probability that $\Sigma^*$ falls into $A$ conditionally on the proposal samples is:
\begin{equation}\label{eq: M role}
\begin{aligned}
    \E\left[\mathbf{1}(\Sigma^* \in A) \mid \{\Sigma^{(m)}\}_{m=1}^M \right]
    &= \sum_{m=1}^{M} \mathbf{1}(\Sigma^{(m)} \in A) \cdot P(\Sigma^* = \Sigma^{(m)}) \\
    &= \sum_{m=1}^{M} \mathbf{1}(\Sigma^{(m)} \in A) \frac{w_m}{\sum_{m=1}^M w_m}\\
    &= \frac{\frac{1}{M}\sum_{m=1}^{M}  w_m \mathbf{1}(\Sigma^{(m)} \in A)}{\frac{1}{M}\sum_{m=1}^{M}  w_m} \\
    &\xrightarrow{M \to \infty} \frac{\E_{\tau}\left[w(\Delta, \Gamma) \mathbf{1}(\Sigma \in A)\right]}{\E_{\tau}\left[w(\Delta, \Gamma)\right]} = \frac{\E_{\tau}\left[\frac{\pi(\Delta, \Gamma)}{\tau (\Delta, \Gamma)}   \mathbf{1}(\Sigma \in A)\right]}{\E_{\tau}\left[\frac{\pi(\Delta, \Gamma)}{\tau (\Delta, \Gamma)}\right]}\\
    &= \E_{\pi}\left[ \mathbf{1}(\Sigma \in A)\right]  \end{aligned}
\end{equation}
where $w_m = w(\Delta^{(m)}, \Gamma^{(m)})$. The convergence comes from the law of large numbers. Therefore, selecting a proposal distribution that is efficient to sample from is key. This allows for a sufficiently large $M$ to ensure the resulting approximation error is negligible, making the final samples a high-fidelity representation of the target distribution in practice. We consider the following proposal, defined in $3$ steps. 
\begin{enumerate}
    \item Draw $(\bm\lambda_1^0, \ldots,  \bm\lambda_K^0, \bm\Gamma_0)$
    \begin{equation}
\begin{cases}
    \bm\Gamma_0 \sim k_{\mathcal{U}\{ \mathcal{O}_K\}}(\Gamma), \\
    \bm\lambda_i^0| \bm\Gamma_0 \sim \pi_{\mathcal{IG}}(\lambda_i |\nu -1, \frac{1}{2}\bm\Gamma_{i}^\top \Psi \bm\Gamma_{i}), \, i = 1, \ldots, K.   
\end{cases}
\end{equation}
    \item Define $I_i, i =1, \ldots, K$ such that $ \bm\lambda_{I_1}^0 \geq \ldots \geq\bm\lambda_{I_K}^0.$
    \item Define matrices $\Tilde{\bm\Delta}$ and $ \Tilde{\bm\Gamma}$:
    \begin{equation}\label{prop}
        \begin{cases}
            \Tilde{\bm\Delta} := \mbox{diag}\{\bm\lambda_{I_1}^0, ..., \bm\lambda_{I_K}^0\}, \\
            \Tilde{\bm\Gamma}_{:,i} = \left[\bm\Gamma_0\right]_{:,I_i}.
        \end{cases}
    \end{equation}
\end{enumerate}
The induced distribution on $(\Tilde{\bm\Delta}, \Tilde{\bm\Gamma})$ is the proposal. The associated weight is: 
\begin{equation}\label{eq: weight}
w(\Delta,\Gamma) \propto \frac{1}{\prod_{i=1}^K c_\mathcal{IG}(\nu -1, \frac{1}{2}\Gamma_{i}^\top \Psi \Gamma_{i})}, \mbox{ where } c_\mathcal{IG}(\nu -1, \frac{1}{2}\Gamma_{i}^\top \Psi \Gamma_{i}) = \frac{\left(\frac{1}{2}\Gamma_{i}^\top \Psi \Gamma_{i}\right)^{\nu-1}}{f_\Gamma(\nu-1)},
\end{equation}
where $f_\Gamma$ is gamma function.
In addition to an easy implementation, another advantage of this proposal is that its weight does not depend on $\Delta$. This means all proposal samples sharing the same $\Gamma$ have the same weight. For importance resampling, the more homogeneous the weights are across proposal samples, the more stable the final result is. By contrast, when there is a large discrepancy in weights, the proposal samples with large importance weights will dominate the output, making it less reliable, even with the consistence always true when $M$ goes to infinity. The extreme case is when one weight is much larger than the others, $N$ resampled samples can consist only of one proposal sample.  We will revisit this issue in next section.  


Plug proposal \eqref{prop} into the SIR framework, we obtain Algorithm \ref{alg: siw}. Note that when $K$ or $\nu$ are large, the weight in Equation \eqref{eq: weight} may encounter overflow or underflow. For example, when $\nu = 50, \Psi = I$, $c_\mathcal{IG}(\nu -1, \frac{1}{2}\Gamma_{i}^\top \Psi \Gamma_{i}) \approx 7\times 10^{75}$. Thus, in practice we work with log-weights $\Log(w)$, and we apply the log-sum-exp trick in the resampling step. 

Now we study the convergence of the output samples of Algorithm \ref{alg: siw} by providing more results on consistency and a central limit theorem. Let us consider any estimator which takes a form of sample mean using the output of Algorithm \ref{alg: siw}, defined as follows.

\begin{definition}\label{def: SIR}
Given the output $\Sigma^{(n)}_M, n = 1, \ldots, N$, of Algorithm \ref{alg: siw}, for any real function $f$ on $S^K_{++}$, define the SIR estimators
\begin{equation}\label{eq: SIR estimator}
\Hat{\mu}^\mathrm{SIR}_{M,N}(f) =
\frac{1}{N}\sum\limits_{n=1}^Nf(\Sigma^{(n)}_M).
\end{equation}
\end{definition}
To be able to express the relative speed between $M$ and $N$, we consider a sequence of values of $(M,N)$, denoted by $\{(M_l,N_l)\}_{l \in \mathbb{N^+}}$, which satisfies $\lim_{l \rightarrow \infty} M_l = \infty$ and $\lim_{l \rightarrow \infty} N_l = \infty$. Then we have following results. 
\begin{theorem}\label{thm: algo1 consistence}
 Assume $f \in L_1(S^K_{++}, \pi)$, we have 
\begin{equation}
    \Hat{\mu}^\mathrm{SIR}_{M_l,N_l}(f) \stackrel{P}{\longrightarrow} \pi(f), \quad l \rightarrow \infty,
\end{equation}
    where $\pi(f)$ denotes the integral $\int_{\Sigma\in S^K_{++}} f(\Sigma) d\pi$.
\end{theorem}
In principle, the consistency holds true as soon as both sample sizes $\ml$ and $\nl$ go to infinity, whichever their relative speed. However, it affects the finite-sample precision of $    \Hat{\mu}^\mathrm{SIR}_{\ml,\nl}$ and the running time of the algorithm. The following result indicates the precision in terms of $\ml/\nl$, therefore guiding the choice of $M, N$ in practice. 
\begin{theorem}\label{thm: algo1 CLT}
Assume $f \in L_2(S^K_{++}, \pi)$, in addition that, $\lim_{l \rightarrow \infty} M_l /N_l = \alpha$ for some $\alpha \in [0, \infty]$, we have when $l \rightarrow \infty,$
 \begin{enumerate}
     \item[(i)] if $\alpha < 1$, then  
\begin{equation}\sqrt{M}\left(\Hat{\mu}^\mathrm{SIR}_{M_l,N_l}(f) -  \pi(f)\right) \stackrel{\mathcal{D}}{\longrightarrow} \mathcal{N}(0, \Tilde{\sigma}^2(f)),
\end{equation}
where $\Tilde{\sigma}^2(f) = \alpha \mathbb{V}_\pi(f) + \mathbb{V}_{\tau}\left\{w [f - \pi(f)]\right\}$ with $\tau$ the proposal defined in Equation \eqref{prop}.
\item[(ii)] if $\alpha \geq 1$, then  
\begin{equation}\sqrt{N}\left(\Hat{\mu}^\mathrm{SIR}_{M_l,N_l}(f) -  \pi(f)\right) \stackrel{\mathcal{D}}{\longrightarrow} \mathcal{N}(0, \Tilde{\sigma}^2(f)), 
\end{equation}
where $\Tilde{\sigma}^2(f) = \mathbb{V}_\pi(f) + \alpha^{-1}\mathbb{V}_{\tau}\left\{w [f - \pi(f)]\right\}$.
\end{enumerate}
\end{theorem}
Firstly, the central limit theorem requires stronger regularity which is a $L_2$-integrability on function $f$ compared to convergence in probability. Secondly, as previously, the theorem holds for any relative speed $\alpha$ of divergence between $M$ and $N$. However, $\alpha$ impacts the convergence rate and the asymptotic variance. The convergence rate is determined by the slower one between $M$ and $N$, while the asymptotic variance is a linear combination of two variances. More specifically, when $\alpha = \infty$, $\Tilde{\sigma}^2(f) = \mathbb{V}_\pi(f)$, which is the asymptotic variance of the sample mean estimator 
\begin{equation}
    \frac{1}{N}\sum\limits_{n=1}^Nf(\Sigma^{(n)}), \quad \Sigma^{(n)} \stackrel{iid}{\sim} \pi.
\end{equation}
This implies when $M$ diverges faster than $N$, the error produced by sampling from the proposal can by ignored, the output can be considered as exact samples from the target, which aligns with the result in calculation \eqref{eq: M role}. On the other hand, when $\alpha = 0$, $\Tilde{\sigma}^2(f) = \mathbb{V}_{\tau}\left\{w [f - \pi(f)]\right\}$, which is the asymptotic variance of the importance sampling estimator 
\begin{equation}
    \frac{\sum_{m=1}^M w_m f(\Tilde{\Sigma}^{(m)})}{\sum_{m=1}^M w_m}, \quad \Tilde{\Sigma}^{(m)} \stackrel{iid}{\sim} \tau.
\end{equation}
In this case, the global estimation error all comes from the error produced by sampling from the proposal. Because the resampling step is faster than the sampling step, in practice, we can always set $N$ larger than $M$, then use $M$ to tune the desired precision. 
\begin{algorithm}[htp]
\caption{Sampling from $\siw(\nu, \Psi, 1)$}\label{alg: siw}
\begin{algorithmic}[1]
\Require degree of freedom $\nu>1$, scale matrix $\Psi\in S^K_{++}$, sample size $N$, sample size for proposal $M$
\Ensure Samples $\{\Sigma^{(n)}_M\}_{n=1}^N \subset S^K_{++}$

\vspace{0.1in}
\textsc{/* Sampling from proposal */}
\For{$m = 1,\dots,M$}
  \State Draw $\alpha_{i,j} \stackrel{iid}{\sim} \mathcal{N}(0, 1), \; i,j =1, \ldots, K,$ and form the matrix $A := [\alpha_{i,j}]$
  \State Compute QR decomposition of $A$, denoted by $A = QR$, and set $\Gamma = Q$
  \For{$i = 1,\dots,K$}
      \State Draw $\lambda_i \sim \mathcal{IG}\bigl(\nu-1,\frac{1}{2}\Gamma_{i}^\top\Psi\Gamma_{i}\bigr)$
  \EndFor
  \State $\Log w_m = K\Log f_\Gamma(\nu-1) - \sum_{i=1}^K(\nu-1)\Log\left(\frac{1}{2}\Gamma_{i}^\top \Psi \Gamma_{i}\right).$
  \State  Set $I_i, i =1, \ldots, K$ such that $ \lambda_{I_1} \geq \ldots \geq \lambda_{I_K}.$
    \State Define matrix $\Tilde{\Delta} := \mbox{diag}\{\lambda_{I_1}, ..., \lambda_{I_K}\}$, and matrix $\Tilde{\Gamma}$ such that $\Tilde{\Gamma}_{:,i} = \left[\Gamma\right]_{:,I_i}$.
  \State Set $\Tilde{\Sigma}^{(m)} \leftarrow \Tilde{\Gamma} \,\Tilde{\Delta}\,\Tilde{\Gamma}^\top$.
\EndFor

\vspace{0.1in}
\textsc{/* Calculating the normalized weights */}
\State 
For $m = 1, \ldots, M$, compute 
$p_m = \Tilde{w}_m/\sum_{m=1}^M \Tilde{w}_m,$ where $\Tilde{w}_m = \Exp(\Log w_m - \max_m\{ \Log w_m \}).$

\vspace{0.1in}
\textsc{/* Importance resampling */}
\For{$n = 1, \ldots, N$}
\State $ m^{\star} \sim \mathcal{M}_M(1 |  p_1, \ldots,  p_M),$ 
\State $\Sigma^{(n)}_M := \Tilde{\Sigma}^{(m^\star)}.  $
\EndFor

\Return $\{\Sigma^{(n)}_M\}_{n=1}^N$
\end{algorithmic}
\end{algorithm}

\noindent


\subsection{Robust version of Algorithm \ref{alg: siw}}\label{sec: algo Psi general robust} 
Even though the proposed weight \eqref{eq: weight} theoretically leads to convergent SIR estimators, in practice, if there is a large discrepancy in weights $w_m$ across proposal samples, it can cause large variance of SIR estimators, $\Hat{\mu}^\mathrm{SIR}_{M,N}(f)$, thus make estimations with finite sample sizes unstable. Moreover, in some extreme cases, the discrepancy can be so large as to make the estimators totally fail due to the machine underflow, illustrated in Figure \ref{fig: log w p}. 

For our weight, large discrepancy can arise, when $\Psi$ has a large discrepancy in its eigenvalues. Recall that the unnormalized weight $w_m$ used in Algorithm \ref{alg: siw} is 
\begin{equation}
    w_m = \prod_{i=1}^K \frac{f_\Gamma(\nu-1)}{\left(\frac{1}{2}\Gamma_{i}^\top \Psi \Gamma_{i}\right)^{\nu-1}}, m = 1, ..., M,
\end{equation}
where $\|\Gamma_{i}\|_{l_2} = 1, \forall \Gamma, i$. The discrepancy in weight is originally from $\Gamma_{i}^\top \Psi \Gamma_{i}$ since $\nu$ is fixed for all samples. However, such discrepancy can be worsen especially by large $\nu$. 

To address this issue, we propose to use SIR with clipped weights. We consider the clipping method proposed in \cite{vazquez2017importance}, which consists in assigning the $M_T$ biggest weights a new value given by the one of the $M_T$-biggest weight. This method is able to reduce the discrepancy of weights. The new SIR framework with clipping is summarized in Algorithm \ref{alg: siw adapted}.
\begin{figure}
    \centering
    \includegraphics[width=0.5\linewidth]{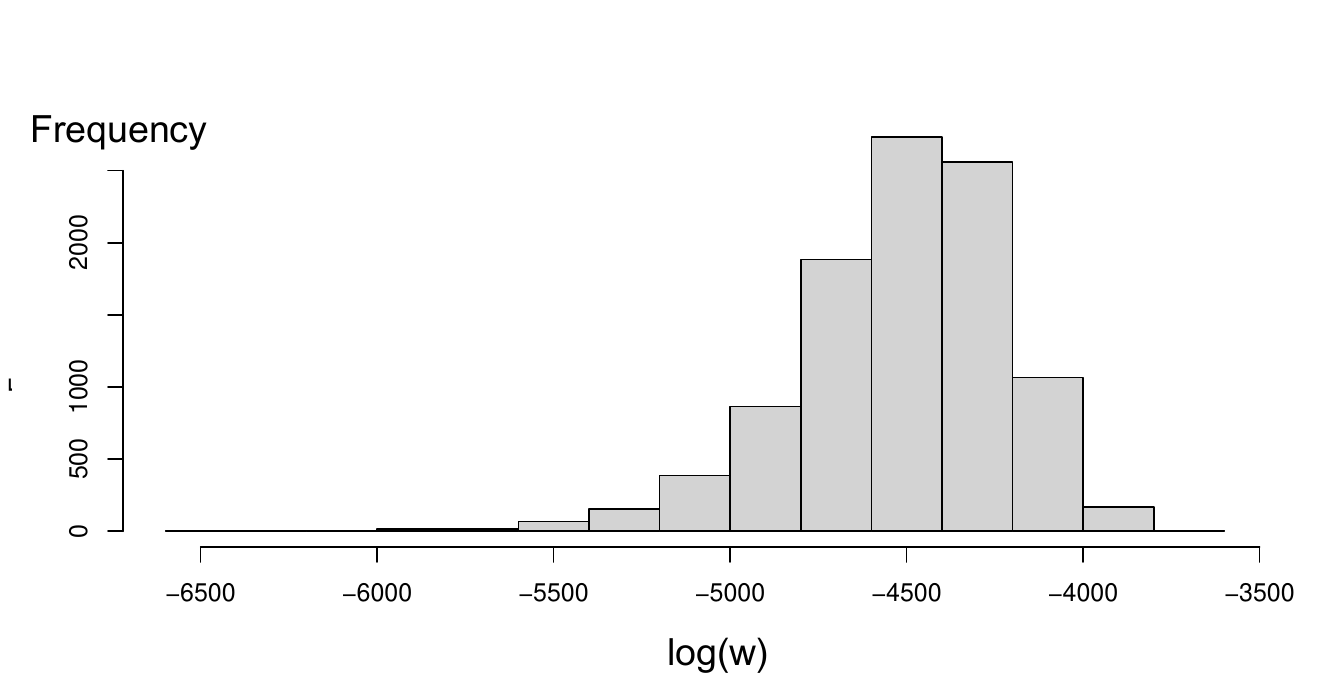}
        \includegraphics[width=0.45\linewidth]{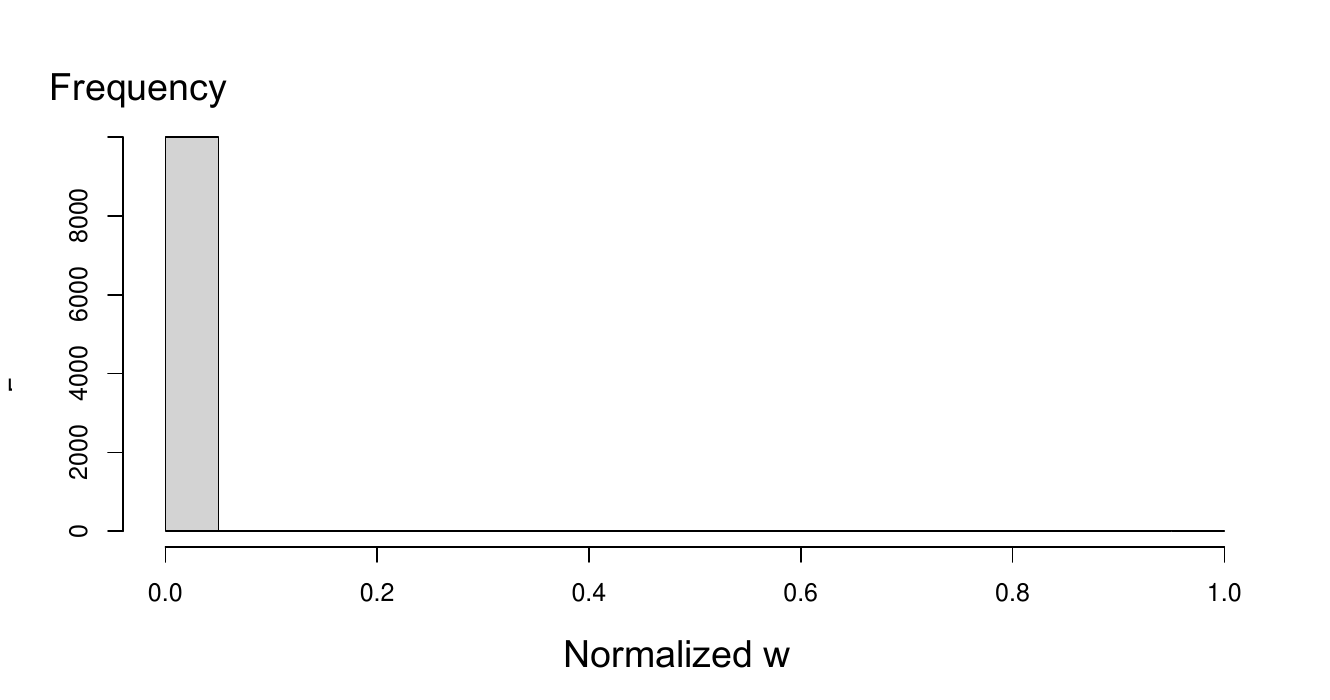}
    \caption{Illustration of the weight collapse phenomenon. 
On the left is a histogram of the log-importance weights ($\log w_m$). Although all weights are finite, their large variance indicates a massive discrepancy in their original values. 
On the right is the corresponding histogram of the normalized weights ($p_m$). Due to the large discrepancy, a single weight dominates, with its normalized value being close to 1 while all others are near 0. Consequently, the resampling step will almost exclusively select one proposal sample, causing the SIR algorithm to fail.}
    \label{fig: log w p}
\end{figure}
\begin{algorithm}
\caption{Robustified sampling from $\siw(\nu, \Psi, 1)$}\label{alg: siw adapted}
\begin{algorithmic}[1]
\Require degree of freedom $\nu>1$, scale matrix $\Psi\in S^K_{++}$, sample size $N$, sample size for proposal $M$, clipping size $M_T$.
\Ensure Samples $\{\Sigma^{(n)}_{\ml,\mtl}\}_{n=1}^N \subset S^K_{++}$

\vspace{0.1in}
\textsc{/* Sampling from proposal */}

\State Same as Algorithm \ref{alg: siw}. 

\vspace{0.1in}
\textsc{/* Clipping */}
\State 
Let $\Log w_{(M_T)}$ the $M_T$-th greatest weight. For $m = 1, \ldots, M$, define  
\vspace{-0.1in}
\begin{equation}
    \begin{cases}
      \Log w_{m}^\prime =  \Log w_{(M_T)}, \; &\mbox{if} \;  \Log w_{m} >  \Log w_{(M_T)}, \\
      \Log w_{m}^\prime = \Log w_{m}, \; &\mbox{otherwise}.
    \end{cases}
\end{equation}

\textsc{/* Calculating the normalized weights */}
\State 
For $m = 1, \ldots, M$, compute 
$p_m = \Tilde{w}_m/\sum_{m=1}^M \Tilde{w}_m,$ where $\Tilde{w}_m = \Exp(\Log w_m^\prime - \max_m\{ \Log w_m^\prime \}).$

\vspace{0.1in}
\textsc{/* Importance resampling */}
\State Same as Algorithm \ref{alg: siw}. 

\Return $\{\Sigma^{(n)}_{\ml,\mtl}\}_{n=1}^N$
\end{algorithmic}
\end{algorithm}
As previously, we shall study the convergence of the output samples of Algorithm \ref{alg: siw adapted}. Let us consider the SIR estimators with clipping defined as follows.
\begin{definition}
Define the SIR estimators with clipping
\begin{equation}
\Hat{\mu}^\mathrm{SIR}_{M,M_T,N}(f) =
\frac{1}{N}\sum\limits_{n=1}^Nf\bigl(\Sigma^{(n)}_{M,M_T}\bigr).
\end{equation}
\end{definition}
Denote a sequence of values of $(M,M_T,N)$ by $\{(M_l,\mtl, N_l)\}_{l \in \mathbb{N^+}}$, which satisfies $\lim_{l \rightarrow \infty} M_l = \infty$, $\lim_{l \rightarrow \infty} \mtl = \infty$ and $\lim_{l \rightarrow \infty} N_l = \infty$. 
\begin{theorem}\label{thm: algo3 consistence}
Assume $f \in L_1(S^K_{++}, \tau)$, and $\lim_{l \rightarrow \infty} M_{T,l}/M_l = 0$, we have 
\begin{equation}
    \Hat{\mu}^\mathrm{SIR}_{M_l,M_{T,l},N_l}(f) \stackrel{P}{\longrightarrow} \pi(f), \quad l \rightarrow \infty.
\end{equation}
\end{theorem}
\begin{theorem}\label{thm: algo3 CLT}
Assume $f \in L_2(S^K_{++}, \tau)$, and $\lim_{l \rightarrow \infty} M_{T,l}/\sqrt{M_l} = 0$, in addition that, $\lim_{l \rightarrow \infty} M_l /N_l = \alpha$ for some $\alpha \in [0, \infty]$, we have when $l \rightarrow \infty$
 \begin{enumerate}
     \item[(i)] if $\alpha < 1$, then  
\begin{equation}\sqrt{\ml}\left(\Hat{\mu}^\mathrm{SIR}_{M_l,M_{T_l},N_l}(f) -  \pi(f)\right) \stackrel{\mathcal{D}}{\longrightarrow} \mathcal{N}(0, \Tilde{\sigma}^2(f)),
\end{equation}
where $\Tilde{\sigma}^2(f) = \alpha \mathbb{V}_\pi(f) + \mathbb{V}_{\tau}\left\{w [f - \pi(f)]\right\}$.
\item[(ii)] if $\alpha \geq 1$, then  
\begin{equation}\sqrt{\nl}\left(\Hat{\mu}^\mathrm{SIR}_{M_l,M_{T_l},N_l}(f) -  \pi(f)\right) \stackrel{\mathcal{D}}{\longrightarrow} \mathcal{N}(0, \Tilde{\sigma}^2(f)), 
\end{equation}
where $\Tilde{\sigma}^2(f) = \mathbb{V}_\pi(f) + \alpha^{-1}\mathbb{V}_{\tau}\left\{w [f - \pi(f)]\right\}$.
\end{enumerate}
\end{theorem}
Theorem \ref{thm: algo3 CLT} shows the convergence rate is determined the same way by $M$ and $N$ as previously. Thus we can still fix $N$ larger than $M$. An important difference in the case of clipping is that, divergence of $M$ and $N$ is not sufficient to reach asymptotic results. The clipped weights should meanwhile not be too many, which is at the magnitude of $o(\ml)$ for the convergence in probability and of $o(\sqrt{\ml})$ for the CLT. While theoretical results suggest that the condition $M_T = o(\sqrt{M_l})$ is necessary to maintain the $\sqrt{M_l}$ convergence rate, our empirical findings indicate this condition can be relaxed in certain cases. A larger clipping size, $M_T$, inherently reduces the variance of the SIR estimator. When this reduction is achieved without degrading the convergence rate, it is clearly advantageous. However, in other scenarios, an excessively large $M_T$ can indeed slow the convergence. This effect can be offset by increasing the number of proposal samples, $M$, albeit at the cost of increased computation time. Consequently, the selection of $M_T$ and $M$ involves a crucial trade-off between the estimator's variance, its convergence rate (which dictates precision), and the overall computational cost. We will detail the way in Section \ref{sec: num} to observe the variance and convergence rate.

Lastly, for the assumption $f \in L_2(S^K_{++}, \tau)$, when $f$ is bounded, the assumption is naturally satisfied. For the more general case of unbounded functions, its validity can be established by leveraging the properties of the $L_2$ space, for which Theorem 2.6 serves as the cornerstone.
\begin{theorem}\label{thm: existence E}
    Let $\bm\Sigma \sim \tau$, with $\tau$ defined in Equation \eqref{prop}. Provided $\Psi \in S^K_{++}$ and $\nu > 1$, $\mathbb{E}\left[\bm\Sigma^{-p}\right]$ exists and is finite, for $p \geq 1$. Additionally, when $p < \nu - 1$,  $\mathbb{E}\left[\bm\Sigma^p\right]$ exists and is finite.
\end{theorem}
We show five examples of application of Theorem \ref{thm: existence E}. 
\begin{exmp}[Matrix entry]\label{ex Sigma ij}
 Let us consider the function $f(\Sigma) = \Sigma_{ij}$, which maps a matrix to one of its scalar entries. For this function to be in $L_2(S^K_{++}, \tau)$, we must verify that $\mathbb{E}[(\Sigma_{ij})^2]$ is finite. Given $\mathbb{E}[(\Sigma_{ij})^2] \leq \mathbb{E}[\sum_{j=1}^K(\Sigma_{ij})^2] = \mathbb{E}[\Sigma^2]_{ii}$, where $[\Sigma^2]_{ii}$ is the $(i,i)$ entry of matrix $\Sigma^2$, $\mathbb{E}[(\Sigma_{ij})^2]$ is finite when the moment of order $p=2$, $\mathbb{E}[\Sigma^2]$, is finite. The later holds true when $p < \nu - 1$, which simplifies to $\nu > 3$. 
\end{exmp}
\begin{exmp}[Trace]
    Consider the trace function, $f(\Sigma) = \mathrm{tr}(\Sigma)$. The trace can be written as, $\mathrm{tr}(\Sigma) = \sum_{i=1}^K \Sigma_{ii}$. The $L_2$ space is a vector space, which implies that any finite linear combination of functions in $L_2$ is also in $L_2$. As established in the previous example, each diagonal entry $\Sigma_{ii}$ is in $L_2$ under the condition $\nu > 3$. By this linearity property, their sum, the trace, is therefore also in $L_2$ under the same condition.
\end{exmp}
\begin{exmp}[Determinant]
Finally, let us examine the determinant function, $f(\Sigma) = \det(\Sigma)$. To verify its $L_2$-integrability, we must check if $\mathbb{E}[(\det(\Sigma))^2]$ is finite. The determinant is a multivariate polynomial of degree $K$ in the matrix entries $\Sigma_{ij}$. Consequently, $(\det(\Sigma))^2$ is a polynomial of degree $2K$. The expectation of such a polynomial is guaranteed to be finite if all moments of the entries $\Sigma_{ij}$ up to order $2K$ exist. According to Theorem 2.5, this condition is met if $\nu > p+1$ for $p=2K$. Thus, the determinant is an $L_2$-integrable function provided that $\nu > 2K + 1$.
\end{exmp}

\begin{exmp}(Inverse square root of determinant)
Let $f(\Sigma) = |\Sigma|^{-1/2}$. Its $L_2$-integrability requires that $\mathbb{E}[|\Sigma|^{-1}]$ is finite. By Hadamard's inequality, $|\Sigma|^{-1} = |\Sigma^{-1}| \le \prod_{i=1}^K (\Sigma^{-1})_{ii}$. The expectation of this product is finite if moments of the entries of $\Sigma^{-1}$ up to order $K$ exist, which is  guaranteed as long as $\nu > 1$. Therefore, $f$ is an $L_2$-integrable function provided that $\nu > 1$. 
\end{exmp}

\begin{exmp}(Gaussian likelihood)
Let $f(\Sigma) = c \cdot |\Sigma|^{-1/2} \exp(-\frac{1}{2}\sum_t(\mathbf{x}_t-\mathbf{u})^T\Sigma^{-1}(\mathbf{x}_t-\mathbf{u}))$, where $c$ is a constant. Its square is proportional to $|\Sigma|^{-1} \exp(-\mathrm{tr}(\Sigma^{-1}S_{\mathbf{x}}))$, where $S_{\mathbf{x}} = \sum_t(\mathbf{x}_t-\mathbf{u})(\mathbf{x}_t-\mathbf{u})^T$ is a positive semi-definite matrix. Since $\mathrm{tr}(\Sigma^{-1}S_{\mathbf{x}}) \ge 0$, the exponential term is bounded by 1. Thus, $(f(\Sigma))^2$ is dominated by a term proportional to $|\Sigma|^{-1} = \det(\Sigma^{-1})$. The integrability of $f$ is therefore guaranteed if $\mathbb{E}[\det(\Sigma^{-1})]$ is finite. As established in the previous example, this holds if $\nu > 1$.
\end{exmp}

\section{Numerical experiments}\label{sec: num}

\subsection{Evaluation of Algorithm \ref{alg: siw Psi=cI}}
As shown in Section \ref{sec: algo Psi simplified}, when $\Psi = cI_{K}, \, c > 0$, we can derive the exact sampling solution without SIR. In this section, we verify the sampling by comparing its sample estimators of moment with their theoretical values provided in \cite{berger2020Bayesian}. More specifically, we consider the estimator of $p$-moment
\begin{equation}\label{eq: sample moments algo2}
    \Hat{m}_p^N(\nu,c) = \frac{1}{N} \sum\limits_{n=1}^N\left[ \Sigma^{(n)}\right]^p,
\end{equation}
where $\Sigma^{(n)}, \, n =1,\ldots, N$ are samples from Algorithm \ref{alg: siw Psi=cI}. We focus on the three moments of $p \in \{-1, 1, 2\}$. Their corresponding true values are 
\begin{equation}
\begin{aligned}
    &m_{-1}(\nu,c) = \mathbb{E}\left[\bm\Sigma^{-1}\right] = \frac{2(\nu-1)}{c}I_K, \; \forall \nu > 1;  \\
    &m_1(\nu,c) = \mathbb{E}\left[\bm\Sigma\right] = \frac{c}{2(\nu-2)}I_K, \; \forall \nu > 2;  \\
    &m_2(\nu,c) = \mathbb{E}\left[\bm\Sigma^{2}\right] = \frac{c^2}{4(\nu-2)(\nu-3)}I_K, \; \forall \nu > 3. 
\end{aligned}
\end{equation}
We consider two values of $\nu$: $4$ and $100$, and fix $c$ $1$. $K \in \{10, 100, 1000\}$, and $N \in \{100, 2100, 3100\}$ are tested. For each combination of $(\nu, c, K, N)$, we run Algorithm \ref{alg: siw Psi=cI} and calculate mean absolute error: 
\begin{equation}
e_p = \left\|\Hat{m}_p^N(\nu,c) - m_p(\nu,c)\right\|_{l_1} = 
\frac{\sum_{i=1}^{K}\sum_{j=1}^{K}\left|[\Hat{m}_p^N(\nu,c)]_{ij} - [m_p(\nu,c)]_{ij}\right|}{K^2},
\end{equation}
for each $p = \{-1, 1, 2\}.$ The results are reported in Table \ref{tbl: e1 e2 e-1}.

\begin{table}[htbp]
  \centering
  \caption{Comparison of empirical errors for the first ($e_1$), second ($e_2$), and inverse ($e_{-1}$) moment estimators against their theoretical values:  $m_1(4,1)=0.25I_K$, $m_1(100,1)=0.0625I_K$, $m_2(4,1)=0.125I_K$, $m_2(100,1)=1.05\times10^{-4}I_K$, $m_{-1}(4,1)=6I_K$, $m_{-1}(100,1)=198I_K$, for different sample sizes $N$.}
  \label{tbl: e1 e2 e-1}
  
  \small
  \setlength{\tabcolsep}{4pt}

  \begin{tabular}{@{}p{2.8cm}l ccc ccc@{}} 
    \toprule
    & & \multicolumn{3}{c}{$\nu = 4, c = 1$} & \multicolumn{3}{c}{$\nu = 100, c = 1$} \\
    \cmidrule(lr){3-5} \cmidrule(lr){6-8}
    Metric & $K$ & 100 & 1100 & 2100 & 100 & 1100 & 2100 \\
    \midrule

    \multirow{3}{=}{\centering \large$e_1$ }
    & 10   & 0.0250 & \colmark{0.0250} & 0.0250 & $5.06\times10^{-4}$ & \colmark{$5.01\times10^{-4}$} & $5.10\times10^{-4}$ \\
    & 100  & 0.0028 & \colmark{0.0025} & 0.0025 & $5.09\times10^{-5}$ & \colmark{$5.10\times10^{-5}$} & $5.10\times10^{-5}$ \\
    & 1000 & $3.13\times10^{-4}$ & \colmark{$2.54\times10^{-4}$} & $2.52\times10^{-4}$ & $5.19\times10^{-6}$ & \colmark{$5.11\times10^{-6}$} & $5.10\times10^{-6}$ \\

    \addlinespace

    \rowe
    \multirow{3}{=}{\centering \large$e_2$ \par \scriptsize $m_2(4,1)=0.125I_K$ \par $m_2(100,1)=1.05\times10^{-4}I_K$}
    & 10   & 0.0125 & \crosscell{0.0125} & 0.0125 & $2.62\times10^{-6}$ & \crosscell{$2.63\times10^{-6}$} & $2.63\times10^{-6}$ \\
    \rowe
    & 100  & 0.0021 & \crosscell{0.0013} & 0.0013 & $2.64\times10^{-7}$ & \crosscell{$2.63\times10^{-7}$} & $2.63\times10^{-7}$ \\
    \rowe
    & 1000 & $3.00\times10^{-4}$ & \crosscell{$1.32\times10^{-4}$} & $1.27\times10^{-4}$ & $2.74\times10^{-8}$ & \crosscell{$2.64\times10^{-8}$} & $2.63\times10^{-8}$ \\

    \addlinespace

    \multirow{3}{=}{\centering \large$e_{-1}$}
    & 10   & 0.6006 & \colmark{0.5999} & 0.5999 & 19.660 & \colmark{19.786} & 19.793 \\
    & 100  & 0.0618 & \colmark{0.0602} & 0.0601 & 1.9760 & \colmark{1.9795} & 1.9797 \\
    & 1000 & 0.0068 & \colmark{0.0061} & 0.0060 & 0.2012 & \colmark{0.1983} & 0.1981 \\
    
    \bottomrule
  \end{tabular}
\end{table}

The tables show that the sample moments in Equation \eqref{eq: sample moments algo2} are very close to the theoretical values even with $N=100$ for all $K$, which validates the correctness of Algorithm \ref{alg: siw adapted}. 
Moreover, we can see that when $K$ goes up the error decreases. This is possibly because the sample moments estimate the off-diagonal entries which are zero better than the diagonal non-zero values. 

In addition, Table \ref{tbl: running time algo Psi = cI} reports running time of Algorithm \ref{alg: siw adapted} in terms of $N$ and $K$. It can be seen that the sampling is highly rapid with less than $4$ seconds for $K=100, N= 2100$. Therefore, it is able to go large dimensions as $K=1000$ which is impossible for the original sampling algorithm proposed in \cite{berger2020Bayesian}. Specifically, the new sampling method proposed in \cite{berger2020Bayesian} is only tested up to $K =100$. Because of the nested Gibbs sampling and the fact that no speed-up is proposed for the special case of $\Psi = cI_K, \, c > 0$, generating one additional sample takes $0.262$ seconds for $K=100$, which is much inferior than our method. 


\begin{table}[htbp]
  \centering
  \caption{Running time of Algorithm 3 in seconds for different parameter settings.}
  \label{tbl: running time algo Psi = cI}
  \begin{tabular}{@{}l ccc ccc@{}}
    \toprule
    & \multicolumn{3}{c}{$\nu = 4, c = 1$} & \multicolumn{3}{c}{$\nu = 100, c = 1$} \\
    \cmidrule(lr){2-4} \cmidrule(lr){5-7}
    $K$ & $N=100$ & $N=1100$ & $N=2100$ & $N=100$ & $N=1100$ & $N=2100$ \\
    \midrule
    10   & 0.005  & 0.054  & 0.091  & 0.0047 & 0.0465 & 0.0958 \\
    100  & 0.170  & 1.885  & 3.615  & 0.1654 & 2.0949 & 3.4536 \\
    1000 & $1.23 \times 10^2$ & $1.36 \times 10^3$ & $2.60 \times 10^3$ & $1.25 \times 10^2$ & $1.36 \times 10^3$ & $2.55 \times 10^3$ \\
    \bottomrule
  \end{tabular}
\end{table} 
\subsection{Evaluation of Algorithm \ref{alg: siw}}\label{sec: num algo siw}
In this section, we test Algorithm \ref{alg: siw}. The main focuses are verifying the theoretical results of convergence and evaluating the discrepancy of weights in terms of $M$ and $K$. 

We firstly present the experiments to verify the convergence of $\Hat{\mu}^\mathrm{SIR}_{M,N}(f)$. To this end, we focus on $f_{ij}(\Sigma) = \Sigma_{ij}, i,j = 1, \ldots K$. Example \ref{ex Sigma ij} shows that when $\nu > 3$, $f \in L_2(S^K_{++}, \pi)$ hence $f \in L_1(S^K_{++}, \pi)$.   Therefore, we set $\nu > 3$ as well in this section to make sure that the convergences in Theorems \ref{thm: algo1 consistence} and \ref{thm: algo1 CLT} hold for $\pi(f_{ij})$ and $\Hat{\mu}^\mathrm{SIR}_{M,N}(f_{ij})$.

Because for the general case $\Psi \neq cI_K$ there is no explicit theoretical results on $\E_{\Sigma\sim \pi}\Sigma$, to verify the convergence of $\Hat{\mu}^\mathrm{SIR}_{M,N}(f_{ij})$, we consider the difference between two realizations of $\Hat{\mu}^\mathrm{SIR}_{M,N}(f_{ij}),$ obtained by running Algorithm \ref{alg: siw} twice with the same parameter and setting values, instead of comparing $\Hat{\mu}^\mathrm{SIR}_{M,N}(f_{ij})$ with $\pi(f_{ij})$. More specifically, for a given $\siw(\nu, \Psi,1)$ and $M, N$, we consider the error metric 
\begin{equation}
e_1^{\mathrm{SIR}} = 
\frac{\sum_{i=1}^{K}\sum_{j=1}^{K}\left|\Hat{\mu}^{\mathrm{SIR},1}_{M,N}(f_{ij}) - \Hat{\mu}^{\mathrm{SIR},2}_{M,N}(f_{ij})\right|}{K^2},
\end{equation}
where $\Hat{\mu}^{\mathrm{SIR},i}_{M,N}(f_{ij}) \stackrel{iid}{\sim} \Hat{\mu}^\mathrm{SIR}_{M,N}(f_{ij}), \; i =1,2$, and $\Hat{\mu}^\mathrm{SIR}_{M,N}(f_{ij})$ is the SIR estimator given in Definition \ref{def: SIR} associated to the given $(\nu,\Psi,N,M)$. When $\Hat{\mu}^\mathrm{SIR}_{M,N}(f_{ij}) \stackrel{P}{\longrightarrow} \mu(f_{ij})$ holds, $e_1^{\mathrm{SIR}} \stackrel{P}{\longrightarrow} 0$. Moreover, when the CLT holds true, there is $\Hat{\mu}^\mathrm{SIR}_{M,N}(f_{ij}) - \mu(f_{ij}) = \mathcal{O}_p(1/\min\{\sqrt{M},\sqrt{N}\})$. It follows that $e_1^{\mathrm{SIR}} = \mathcal{O}_p(1/\min\{\sqrt{M},\sqrt{N}\})$. Therefore, we rely on the converging behaviors of curves $e_1^{\mathrm{SIR}}$ to verify the theoretical results. 

We test $K =10, 100$. Because in practice, we suggest to always set $N$ larger than $M$, in experiments, we fix $N = 5M$. Thus, the convergence rate is given by $M$. We consider the sequence of $M = 500, 1000, 1500, \ldots, 10000$. $\nu$ is set as $4$ and $20$. For each combination of $(K, \nu)$, we generate $2$ $\Psi$'s respectively of large and small discrepancies in eigenvalues. More specifically, we first draw two orthonormal matrices from $\pi_{\mathcal{O}_K}$. Then we generate the eigenvalues of small discrepancy as
\begin{equation}
    \mbox{Case }1: \{2, 1.01\} \cup \{\lambda_k+1, k = 2, ..., K-1\}, \quad \lambda_k \stackrel{iid}{\sim} \mathcal{U}(0.01,1),
\end{equation}
and the ones of large discrepancy as 
\begin{equation}
    \mbox{Case }2: \{1, 1.01\} \cup \{\lambda_k, k = 2, ..., K-1\}, \quad \lambda_k \stackrel{iid}{\sim} \mathcal{U}(0.01,1).
\end{equation}
For each combination of $(K, \nu, \Psi)$, we calculate $e_1^{\mathrm{SIR}}$ for all $M$ in the sequence in order to draw a curve. To furthermore study the stability of SIR estimator, we create $10$ such curves by running $10$ independent simulations for the same $\Psi$. 
First, we show the convergence results by the curves of $e_1^{\mathrm{SIR}}$ in Figures \ref{fig: case 1} and \ref{fig: case 2}. 

Figure \ref{fig: 10, 4, i} shows the evolutions of $e_1^{\mathrm{SIR}}$ and $\sqrt{M}e_1^{\mathrm{SIR}}$ in terms of proposal sample size $M$, respectively on the left and right, corresponding to $K=10, \nu =4$ and small eigen-discrepancy. All curves decrease towards zero. Moreover, the sample standard deviation represented by the blue shade is decreasing as well, which implies a $L_2$ convergence hence convergence in probability to zero of the distance of two SIR estimators $e_1^{\mathrm{SIR}}$. Furthermore, we would like to investigate the convergence rate. On the right subfigure of Figure \ref{fig: 10, 4, i}, it can be seen that all curves of $\sqrt{M}e_1^{\mathrm{SIR}}$ are bounded, which is consistent with the expected result that $e_1^{\mathrm{SIR}} = \mathcal{O}_p(1/\sqrt{M})$. Thus it validates numerically that $\Hat{\mu}^\mathrm{SIR}_{M,N}(f_{ij}) - \mu(f_{ij}) = \mathcal{O}_p(1/\sqrt{M})$, which is a consequence of the derived CLT. We recall that we set $N > M$, thus the convergence rate in CLT should be given by $\sqrt{M}$.

When $\nu$ increases from $4$ to $20$, similar convergence patterns can be found in Figure~\ref{fig: 10, 20, i}, except that the sample standard deviation does not decrease as $M$ goes up. Even though this is not unexpected since we did not provide $L_2$ convergence. It can still reveal the fact that the SIR estimators perform less well when $\nu$ increases, which accentuates the discrepancy in weights. This can be seen also in Table \ref{tbl:ess_algo siw}, with slightly reduced effective sample size (ESS) on Line $(10,20,\mbox{Case } 1)$. 

As $K$ increases to $100$, the SIR estimators perform better in Figures \ref{fig: 100, 4, i} and \ref{fig: 100, 20, i} with smaller sample standard deviations, as well as in Table \ref{tbl: e1} with higher ESS's. This is not counter-intuitive because we intentionally retained the distance between $\lambda_{\max}$ and $\lambda_{\min}$ at the same level as $K=10$. The fact that there are more eigenvalues within the same range of $\lambda_{\max}$ and $\lambda_{\min}$ for $K=100$ can reduce the overall weight discrepancy.

\begin{figure}[htbp!]
\centering 
\begin{subfigure}[b]{\textwidth}
\centering
\includegraphics[width=0.45\linewidth]{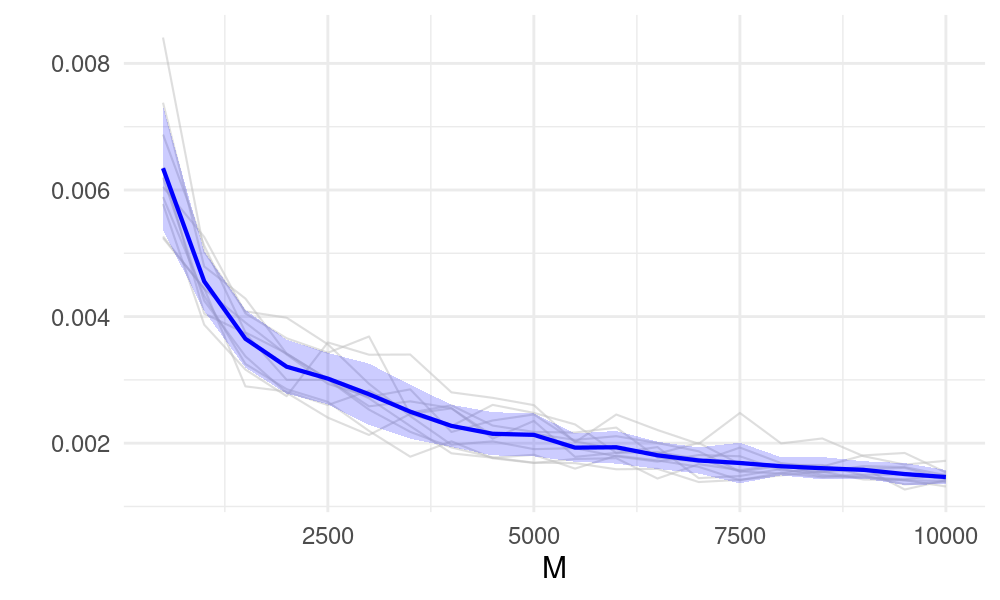} 
\includegraphics[width=0.45\linewidth]{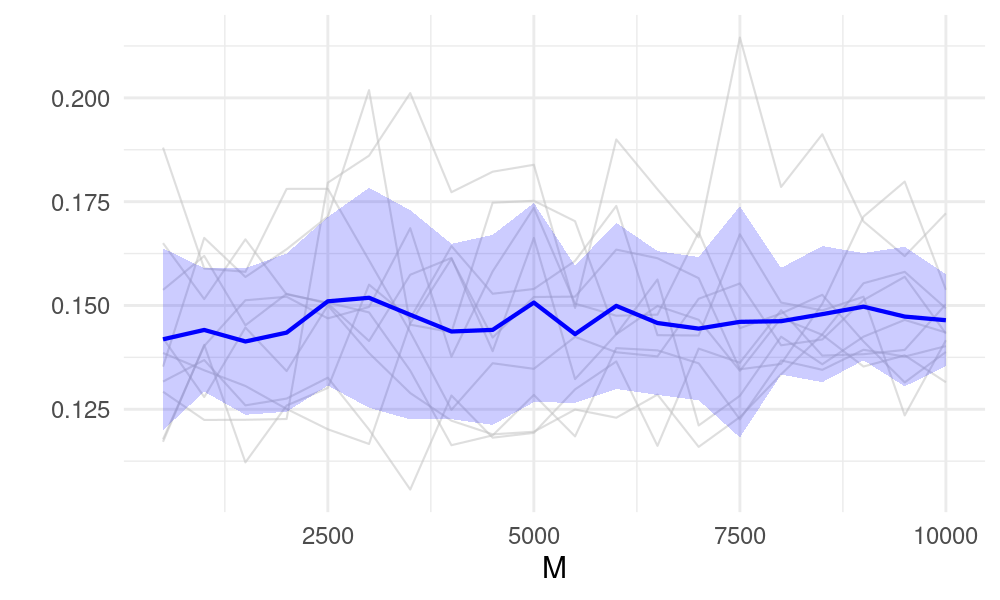} 
\subcaption{
$K=10, \nu=4$.
}
\label{fig: 10, 4, i}
\end{subfigure}
\begin{subfigure}[b]{\textwidth}
\centering
\includegraphics[width=0.45\linewidth]{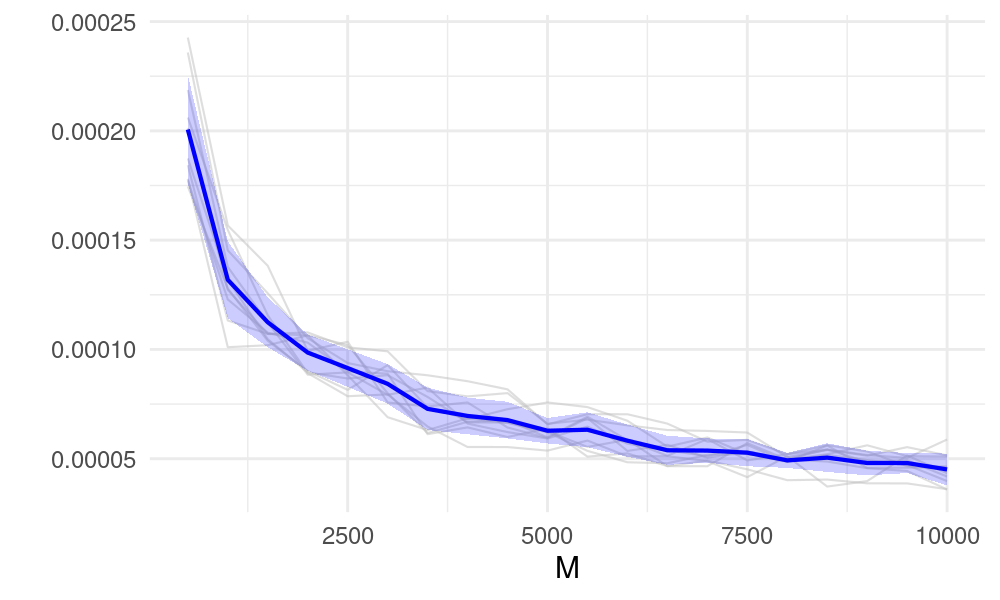} 
\includegraphics[width=0.45\linewidth]{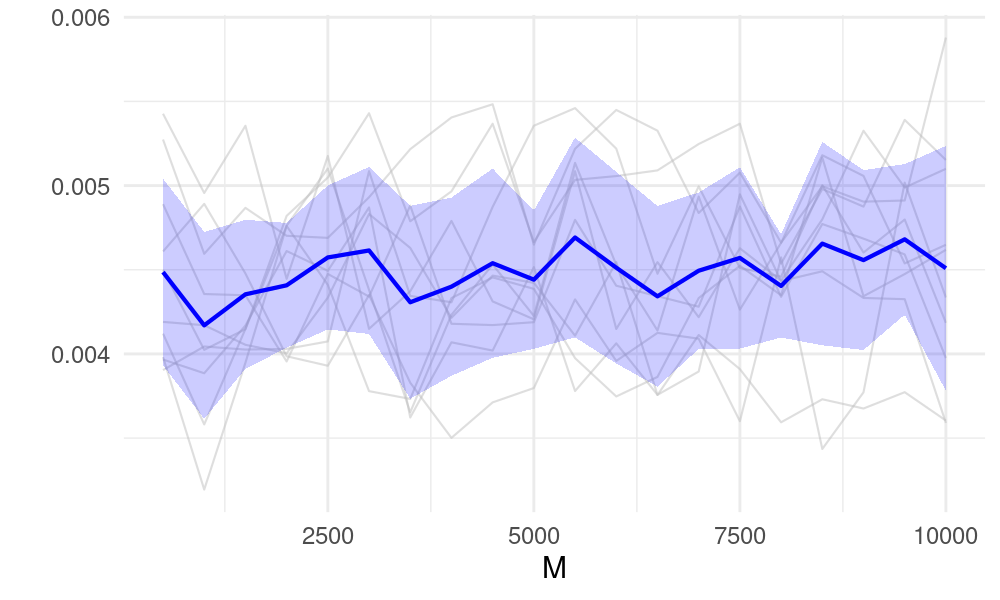} 
\subcaption{
$K=10, \nu=20$.
}
\label{fig: 10, 20, i}
\end{subfigure}
\begin{subfigure}[b]{\textwidth}
\centering
\includegraphics[width=0.45\linewidth]{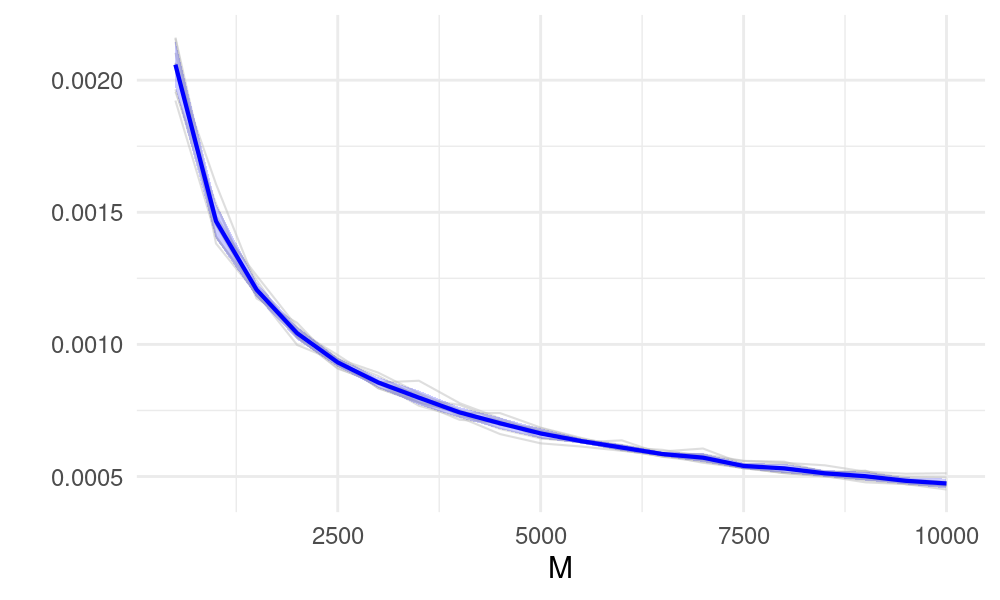} 
\includegraphics[width=0.45\linewidth]{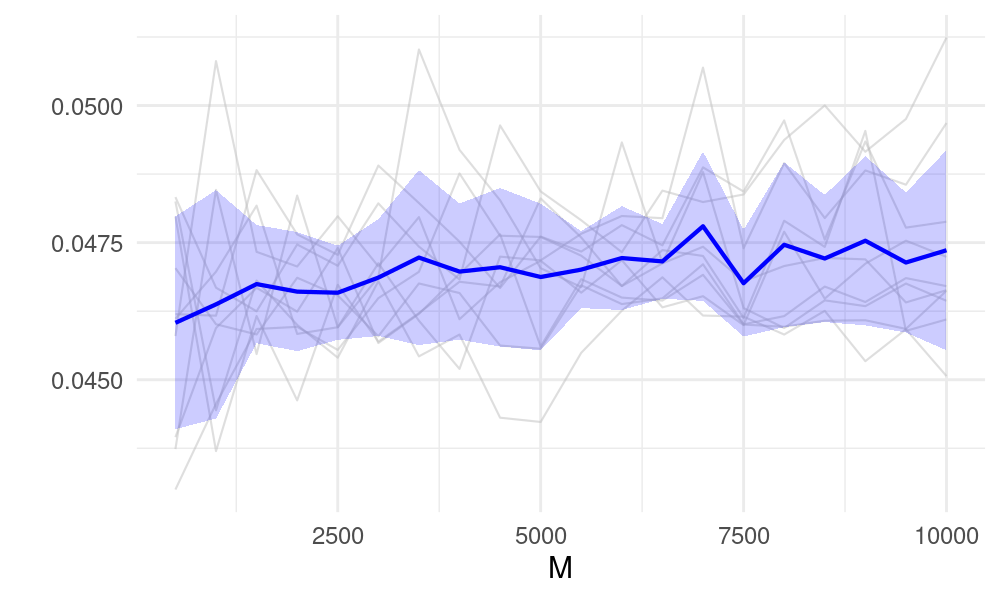} 
\subcaption{
$K=100, \nu=4$.
}
\label{fig: 100, 4, i}
\end{subfigure}
\begin{subfigure}[b]{\textwidth}
\centering
\includegraphics[width=0.45\linewidth]{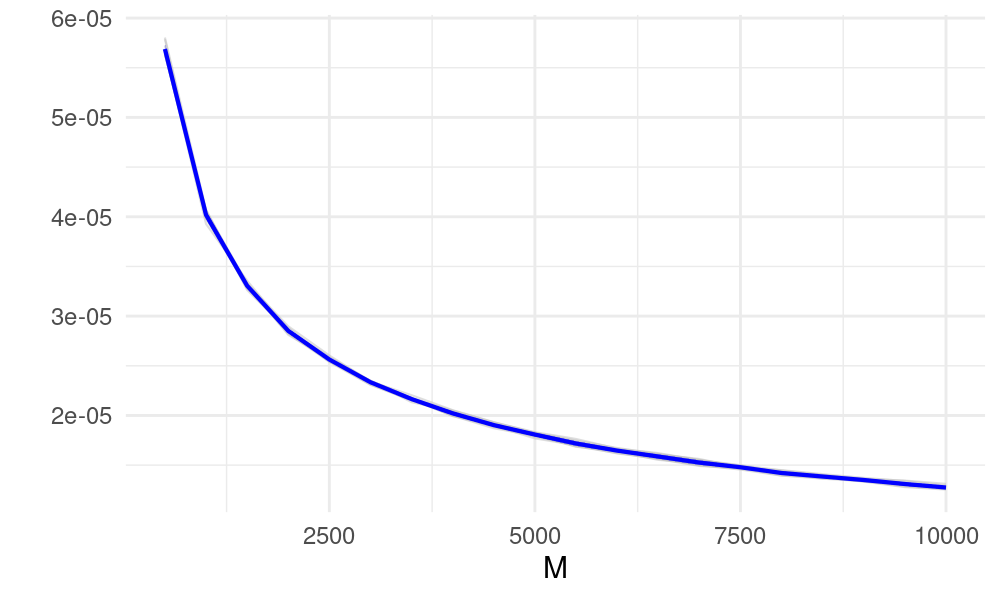} 
\includegraphics[width=0.45\linewidth]{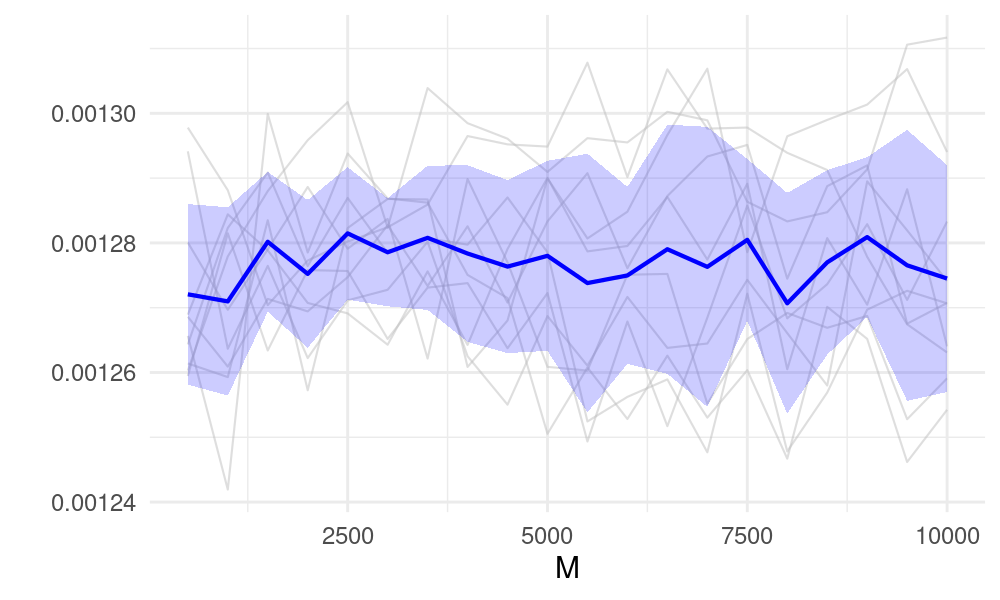} 
\subcaption{
$K=100, \nu=20$.
}
\label{fig: 100, 20, i}
\end{subfigure}
\caption{Evolutions of $e_1^{\mathrm{SIR}}$ (left) and $\sqrt{M}e_1^{\mathrm{SIR}}$ (right) in terms of $M$ for the eigen-discrepancy of \textbf{Case 1}. Gray curves represent independent simulations. Blue curves represent their means and blue shades represent their standard deviation.}
\label{fig: case 1}
\end{figure}

When eigen-discrepancy increases with pattern $2$, we can see in Table \ref{tbl: e1} that ESS decreases in all settings of $K,\nu$, more significantly with $\nu=20$. In the same table, a weight collapse is observed with $K=10,\nu=20$. In this case, a very few weights dominate, thus the resampling samples consist in only several unique samples, which cause the failures in Figure \ref{fig: 10, 20, ii}. In the next section, we will test again these same settings with the same generated $\Psi$'s using Algorithm \ref{alg: siw adapted}. Before proceeding to the improved results, we lastly report the running time of Algorithm \ref{alg: siw} in Table \ref{tbl: runing time algo Psi general}. 

\begin{figure}[htbp!]
\centering 
\begin{subfigure}[b]{\textwidth}
\centering
\includegraphics[width=0.45\linewidth]{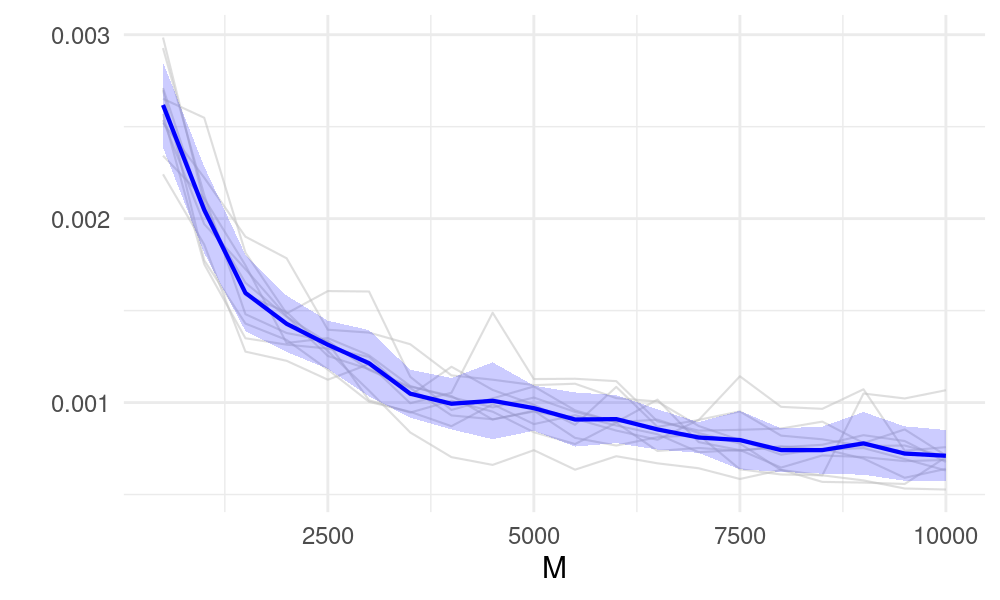} 
\includegraphics[width=0.45\linewidth]{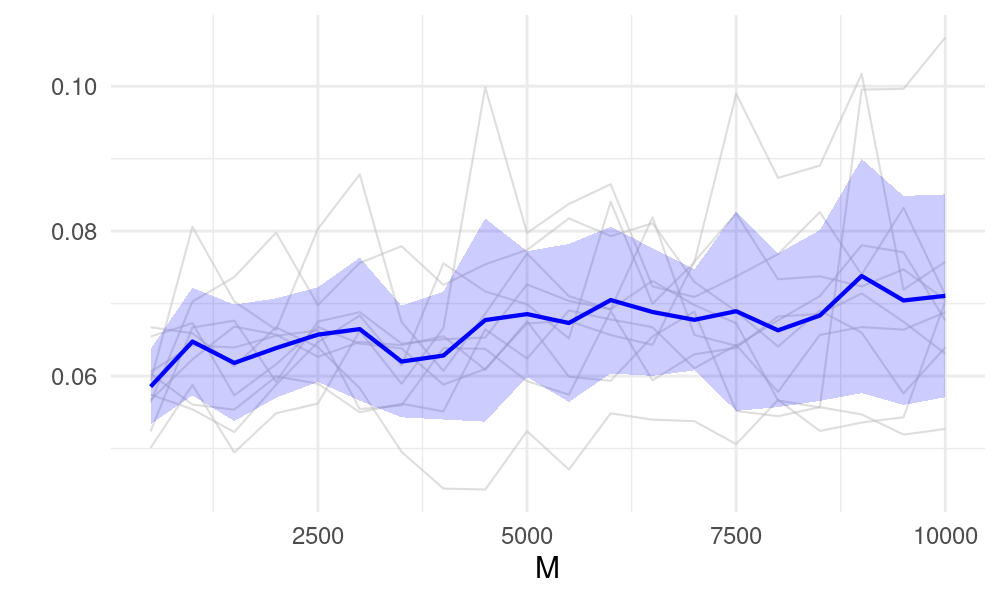} 
\subcaption{
$K=10, \nu=4$.
}
\label{fig: 10, 4, ii}
\end{subfigure}
\begin{subfigure}[b]{\textwidth}
\centering
\includegraphics[width=0.45\linewidth]{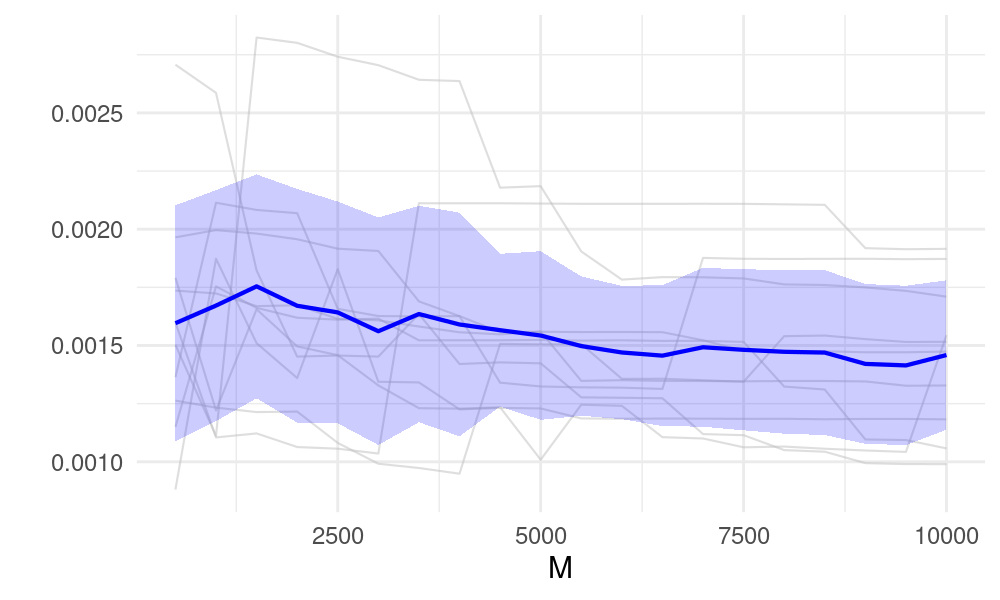} 
\includegraphics[width=0.45\linewidth]{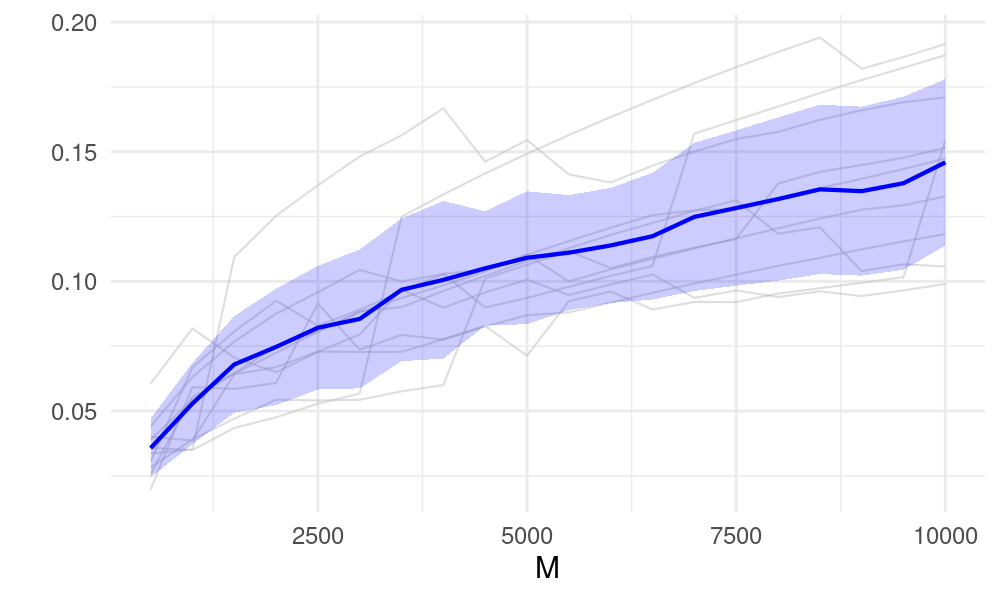} 
\subcaption{
$K=10, \nu=20$.
}
\label{fig: 10, 20, ii}
\end{subfigure}
\begin{subfigure}[b]{\textwidth}
\centering
\includegraphics[width=0.45\linewidth]{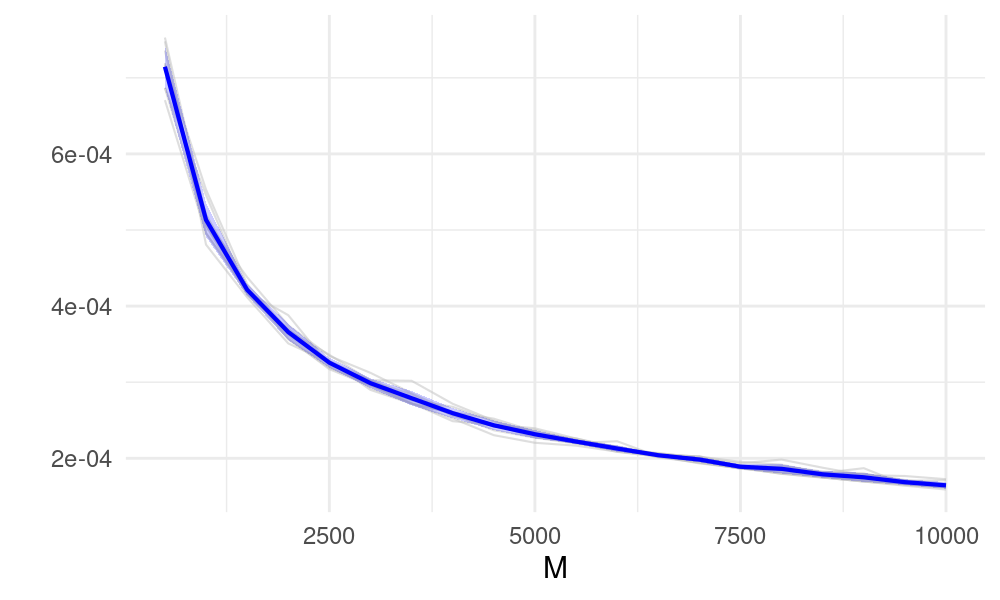} 
\includegraphics[width=0.45\linewidth]{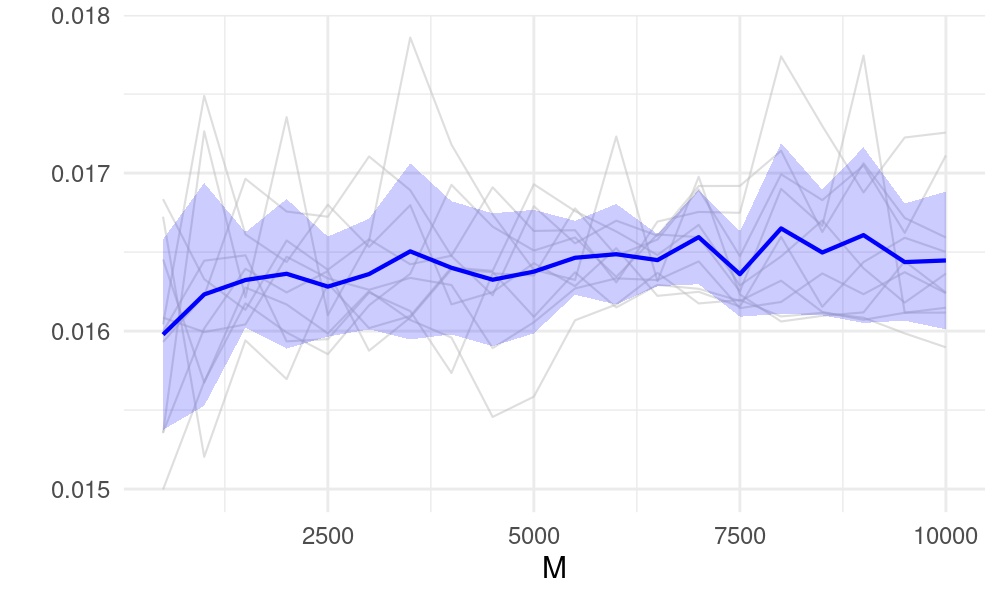} 
\subcaption{
$K=100, \nu=4$.
}
\label{fig: 100, 4, ii}
\end{subfigure}
\begin{subfigure}[b]{\textwidth}
\centering
\includegraphics[width=0.45\linewidth]{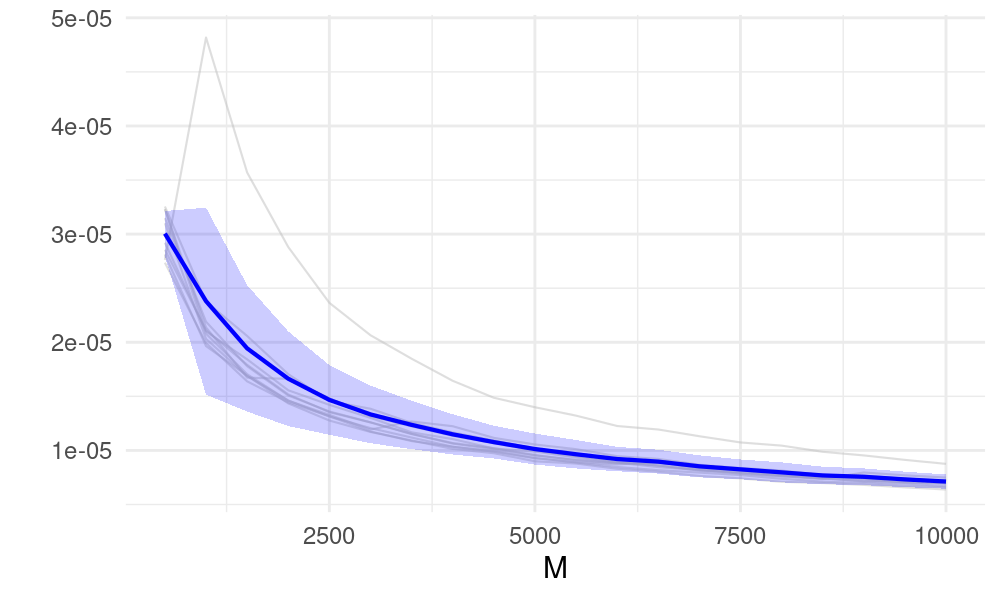} 
\includegraphics[width=0.45\linewidth]{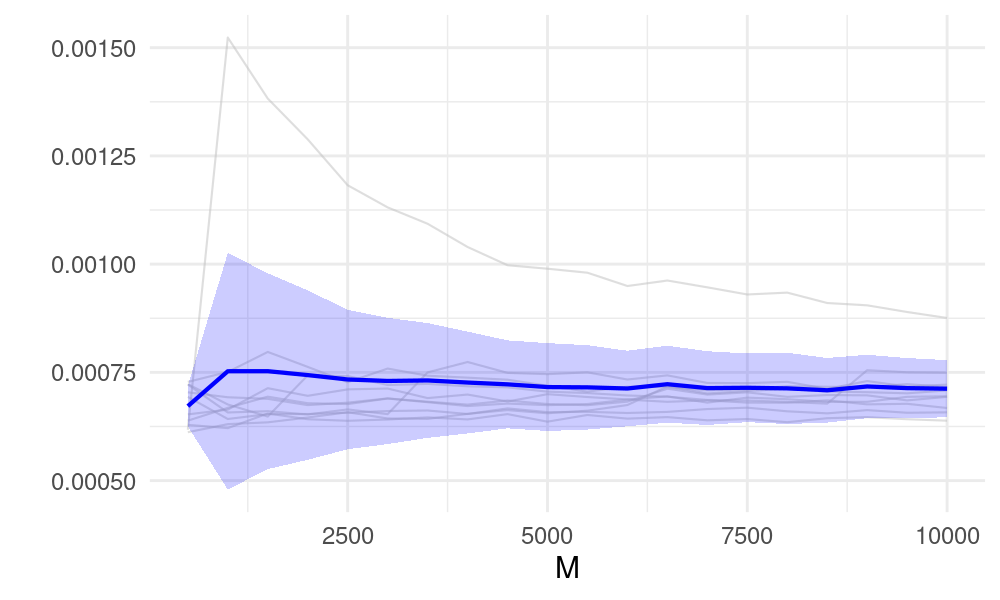} 
\subcaption{
$K=100, \nu=20$.
}
\label{fig: 100, 20, ii}
\end{subfigure}
\caption{Evolutions of $e_1^{\mathrm{SIR}}$ (left) and $\sqrt{M}e_1^{\mathrm{SIR}}$ (right) in terms of $M$ for the eigen-discrepancy of \textbf{Case 2}. }
\label{fig: case 2}
\end{figure}

\definecolor{rowhi}{RGB}{245,245,245}     
\definecolor{mixcolor}{RGB}{235,235,235}  
\definecolor{colhi}{RGB}{235,246,252}     

\colorlet{mix}{rowhi!50!colhi}  

\newcolumntype{H}{>{\columncolor{colhi}}c}  
\newcolumntype{N}{c}                        

\newcommand{\mixcell}[1]{\cellcolor{mixcolor}#1}

\begin{table}[htbp!]
  \centering
  \caption{Mean ESS (in \% of $M$) over 10 simulations of Algorithm \ref{alg: siw}.}
  \label{tbl:ess_algo siw}

  \begin{tabular}{@{}l N H N H N H@{}}
    \toprule
    \rowcolor{white}
    \diagbox[height=2.5em, width=8em]{$K,\nu,\Psi$}{$M$}
      & 500 & 2500 & 4500 & 6500 & 8500 & 10000 \\
    \midrule

    (10, 4, Case 1)   & 99.8\% & 99.8\% & 99.8\% & 99.8\% & 99.8\% & 99.8\% \\
    (100, 4, Case 1)  & 100.0\% & 100.0\% & 100.0\% & 100.0\% & 100.0\% & 100.0\% \\
    (10, 20, Case 1)  & 92.0\% & 92.0\% & 92.0\% & 91.9\% & 91.9\% & 91.9\% \\
    (100, 20, Case 1) & 99.0\% & 99.1\% & 99.1\% & 99.1\% & 99.1\% & 99.1\% \\

    \rowcolor{rowhi}
    (10, 4, Case 2)   & 58.2\% & \mixcell{54.8\%} & 56.4\% & \mixcell{57.1\%} & 55.4\% & \mixcell{54.4\%} \\
    \rowcolor{rowhi}
    (100, 4, Case 2)  & 98.2\% & \mixcell{98.3\%} & 98.3\% & \mixcell{98.3\%} & 98.3\% & \mixcell{98.3\%} \\
    \rowcolor{rowhi}
    (10, 20, Case 2)  & 0.4\%  & \mixcell{0.1\%}  & 0.0\%  & \mixcell{0.0\%}  & 0.0\%  & \mixcell{0.0\%} \\
    \rowcolor{rowhi}
    (100, 20, Case 2) & 43.8\% & \mixcell{38.5\%} & 37.6\% & \mixcell{37.8\%} & 38.6\% & \mixcell{38.6\%} \\

    \bottomrule
  \end{tabular}
\end{table}

\definecolor{rowhi}{RGB}{245,245,245}     
\definecolor{mixcolor}{RGB}{235,235,235}  
\definecolor{colhi}{RGB}{235,246,252}     

\begin{table}[htbp!]
  \centering
  \caption{Mean running time over $10$ simulations (in seconds) of Algorithm \ref{alg: siw}.}
  \label{tbl: runing time algo Psi general}

  \begin{tabular}{@{}l N H N H N H@{}}
    \toprule
    \rowcolor{white}
    \diagbox[height=2.5em, width=8em]{$K,\nu,\Psi$}{$M$}
      & 500 & 2500 & 4500 & 6500 & 8500 & 10000 \\
    \midrule

    (10, 4, Case 1)   & 0.053 & 0.208 & 0.377 & 0.528 & 0.684 & 0.800 \\
    (100, 4, Case 1)  & 1.494 & 6.512 & 11.522 & 16.540 & 21.576 & 25.361 \\
    (10, 20, Case 1)  & 0.039 & 0.199 & 0.357 & 0.511 & 0.668 & 0.789 \\
    (100, 20, Case 1) & 1.456 & 6.426 & 11.418 & 16.392 & 21.353 & 25.081 \\

    \rowcolor{rowhi}
    (10, 4, Case 2)   & 0.039 & \mixcell{0.198} & 0.342 & \mixcell{0.487} & 0.631 & \mixcell{0.739} \\
    \rowcolor{rowhi}
    (100, 4, Case 2)  & 1.463 & \mixcell{6.534} & 11.599 & \mixcell{16.654} & 21.688 & \mixcell{25.466} \\
    \rowcolor{rowhi}
    (10, 20, Case 2)  & 0.041 & \mixcell{0.186} & 0.330 & \mixcell{0.475} & 0.621 & \mixcell{0.731} \\
    \rowcolor{rowhi}
    (100, 20, Case 2) & 1.439 & \mixcell{6.413} & 11.407 & \mixcell{16.397} & 21.385 & \mixcell{25.128} \\

    \bottomrule
  \end{tabular}
\end{table}

We only record the running time of sampling from proposal for the convenience of organizing experiments. In practice, the steps of calculating the normalized weights and importance resampling take very little time. Table \ref{tbl: runing time algo Psi general} shows that the running time of Algorithm \ref{alg: siw} is still very fast even though slightly slower than Algorithm \ref{alg: siw Psi=cI}. In addition, no significant difference is observed when $\nu$ or $\Psi$ change. The main factors that impact the running time is $K$ and $M$. 
We recall that generating one sample using the nested Gibbs sampling in \cite{berger2020Bayesian} takes 0.262 seconds for $K=100$. In comparison, generating $500$ ``effective" samples, which are samples from proposal to be resampled, using our Algorithm \ref{alg: siw} for $K=100$ needs only around 1.5 seconds, thus $0.003$ seconds per proposal sample. Based on the same set of proposal samples ($M$ fixed), we can obtain a much larger set of output samples ($N \gg M$) with almost no addition time, since resampling only performs sampling from a Multinomial distribution.

\subsection{Evaluation of Algorithm \ref{alg: siw adapted}}\label{sec: eva algo adapted}
In this section, we focus on study the improvement of estimator robustness that brought by the clipping. Additionally, we investigate the impacts of different clipping sizes on the robustness and the convergence of the SIR estimators. To this end, we test $3$ clipping sizes, $M_T = M^{0.2}, M^{0.45}, M^{0.8}$. The other settings are the same as in Section \ref{sec: num algo siw}, except we only consider large eigen-discrepancy. Firstly, the new ESS's are reported in Table \ref{tab:ess_clipping}. As more weights are clipped, more increase is observed in ESS, implying more robustness in the SIR estimators, with respect to Table \ref{tbl:ess_algo siw}. However for each fixed $M_T$, as $M$ increases, the efficiency of clipping declines slightly. 
\definecolor{rowhi}{RGB}{245,245,245}   
\definecolor{rowhi2}{RGB}{245,240,240}   %
\definecolor{mixcolor}{RGB}{235,235,235}   
\definecolor{mixcolor2}{RGB}{223,220,220}   %
\definecolor{colhi}{RGB}{235,246,252}   
\colorlet{mix}{rowhi!50!colhi}          

\newcommand{\mixcelll}[1]{\cellcolor{mixcolor2}#1}

\begin{table}[htbp!]
  \centering
  \caption{Mean ESS (in $\%$ of $M$) for different clipping sizes $M_T$. The results correspond to the \textbf{large} eigen-discrepancy case (\textbf{Case 2}). }
  \label{tab:ess_clipping}

  \begin{tabular}{@{}l N H N N H N N H N@{}}
    \toprule
    \rowcolor{white}
    & \multicolumn{3}{c}{$M_T = M^{0.2}$}
    & \multicolumn{3}{c}{$M_T = M^{0.45}$}
    & \multicolumn{3}{c}{$M_T = M^{0.8}$} \\
    \cmidrule(lr){2-4} \cmidrule(lr){5-7} \cmidrule(lr){8-10}
    \rowcolor{white}
    \diagbox[height=2.5em, width=8em]{$K,\nu$}{$M$}
    & 2500 & 6500 & 10000 & 2500 & 6500 & 10000 & 2500 & 6500 & 10000 \\
    \midrule

    \rowcolor{rowhi}
    (10, 4)
      & 66.8\% & \mixcell{64.1\%} & 62.2\%
      & 78.2\% & \mixcell{75.5\%} & 73.6\%
      & 93.2\% & \mixcell{92.2\%} & 91.7\% \\

    (100, 4)
      & 98.3\% & 98.3\% & 98.3\%
      & 98.4\% & 98.3\% & 98.3\%
      & 99.0\% & 98.9\% & 98.9\% \\

    \rowcolor{rowhi2}
    (10, 20)
      & 0.2\% & \mixcelll{0.1\%} & 0.1\%
      & 3.0\% & \mixcelll{1.7\%} & 1.4\%
      & 41.1\% & \mixcelll{35.3\%} & 33.0\% \\

    (100, 20)
      & 46.0\% & 43.5\% & 43.4\%
      & 52.9\% & 50.7\% & 49.8\%
      & 77.2\% & 74.8\% & 73.6\% \\
    \bottomrule
  \end{tabular}
\end{table}

Secondly, we report the evolutions of $e_1^{\mathrm{SIR}}$. In this section, we only show the results of $K=10$ with the ones of $K=100$ in the appendices in order to avoid redundant remarks. Firstly, Figures \ref{fig: 10, 20, Case 2 MT20}, \ref{fig: 10, 20, Case 2 MT45}, and \ref{fig: 10, 20, Case 2 MT80} report the results on $K=10,\nu=20$ with large eigen-discrepancy. This parameter setting has a weight collapse. We can see, even though $M_T = M^{0.2}$ is not enough to bring back a reasonable ESS, the curves of $e_1^{\mathrm{SIR}}$ start to show already the convergence pattern, as displayed in the left subfigure of Figure \ref{fig: 10, 20, Case 2 MT20}. However the corresponding curves of $\sqrt{M}e_1^{\mathrm{SIR}}$ seem unbounded. When $M_T = M^{0.45}$, convergence and boundedness both show. However the ESS is still too low to have robust estimations. We then test $M_T = M^{0.8}$. The clipping size is larger than the maximal size required by the CLT which is $M^{0.5}$, thus there is no theoretical guarantee on the convergence rate. However it is less than $M$, thus the consistency in probability still holds. We can see in Figure \ref{fig: 10, 20, Case 2 MT80} the corresponding convergence pattern. Thanks to the relatively large number of clipped weights, the standard deviation is much smaller, implying robust estimation even with large sample size. Moreover, on the right, the boundedness is also present even without theoretical results. Thus for this tricky case of $K=10, \nu=20$ with large eigen-discrepancy, we suggest to use large clipping size greater than $M^{0.5}$. 
\begin{figure}[htbp]
\centering 
\begin{subfigure}[b]{\textwidth}
\centering
\includegraphics[width=0.45\linewidth]{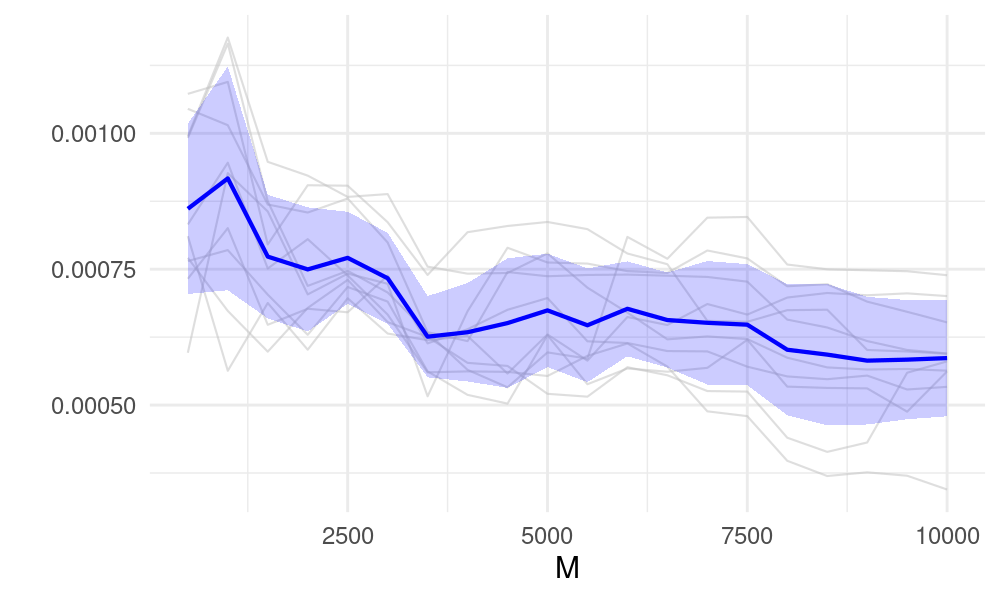} 
\includegraphics[width=0.45\linewidth]{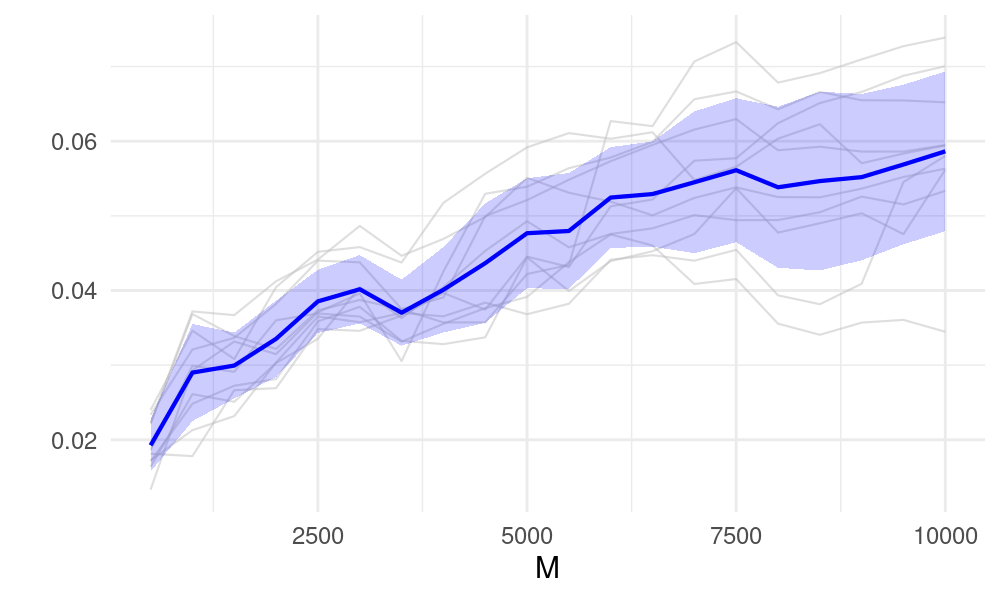} 
\subcaption{
$M_T = M^{0.2}$.
}
\label{fig: 10, 20, Case 2 MT20}
\end{subfigure}
\begin{subfigure}[b]{\textwidth}
\centering
\includegraphics[width=0.45\linewidth]{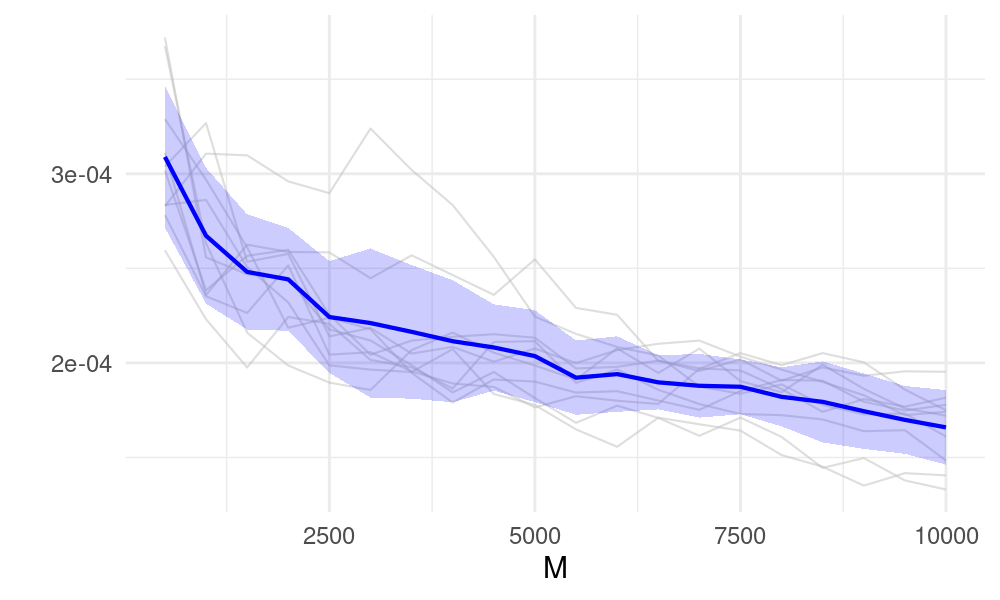} 
\includegraphics[width=0.45\linewidth]{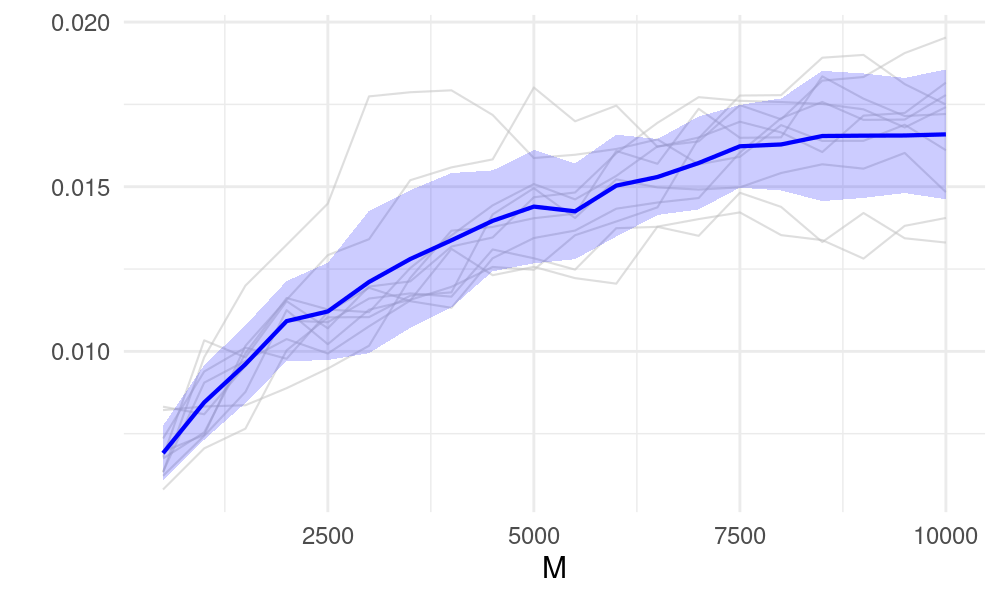}
\subcaption{
$M_T = M^{0.45}$.
}
\label{fig: 10, 20, Case 2 MT45}
\end{subfigure}
\begin{subfigure}[b]{\textwidth}
\centering
\includegraphics[width=0.45\linewidth]{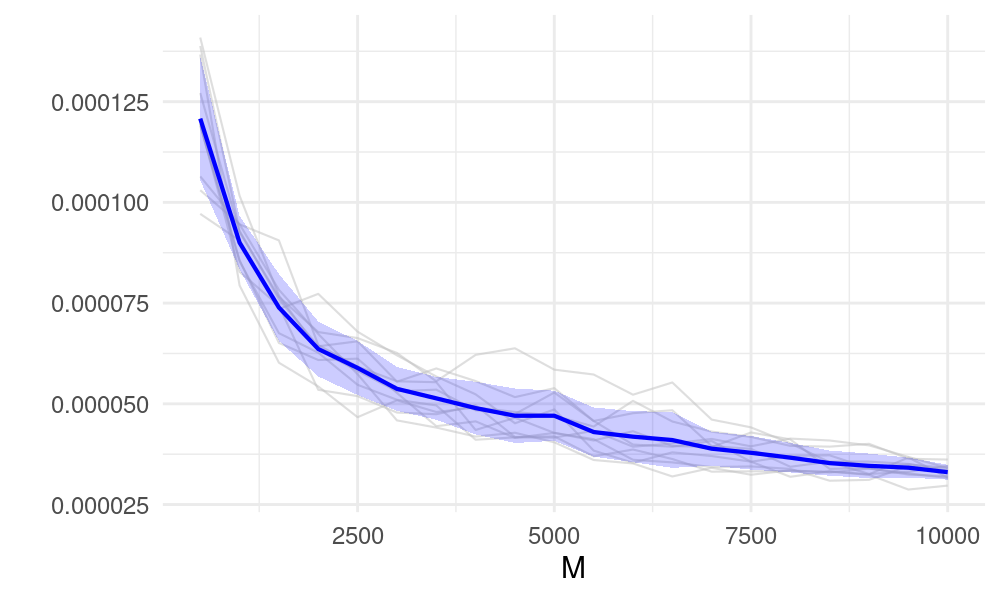} 
\includegraphics[width=0.45\linewidth]{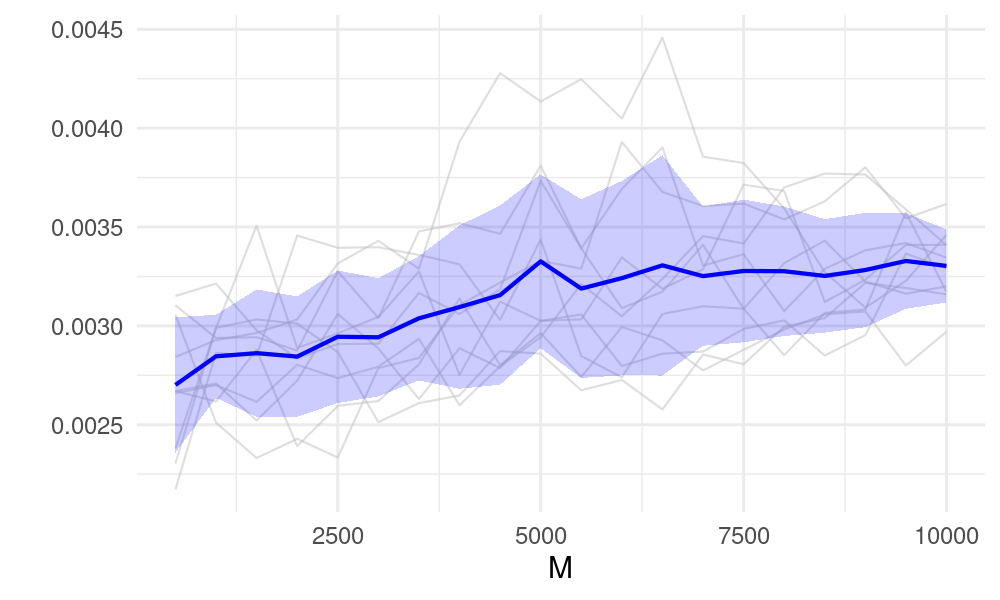} 
\subcaption{
$M_T = M^{0.8}$.
}
\label{fig: 10, 20, Case 2 MT80}
\end{subfigure}
\caption{Evolutions of $e_1^{\mathrm{SIR}}$ (left) and $\sqrt{M}e_1^{\mathrm{SIR}}$ (right) in terms of $M$ for $K=10, \nu=20, \mbox{eigen-discrepancy of Case 2}$, different clipping sizes. }
\label{fig: case 2, K=10, nu=20}
\end{figure}
\begin{figure}[htbp]
\centering 
\begin{subfigure}[b]{\textwidth}
\centering
\includegraphics[width=0.45\linewidth]{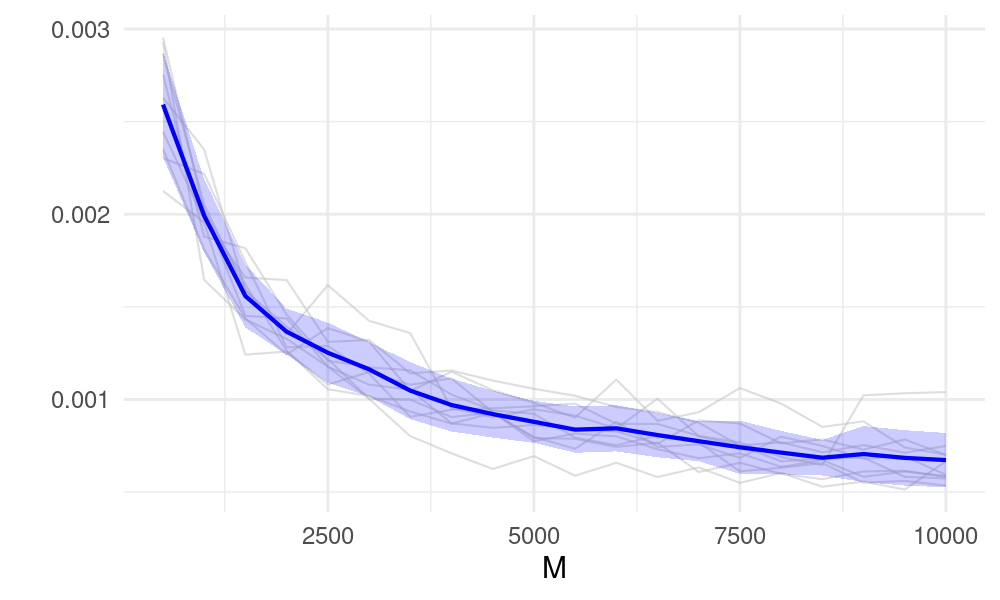} 
\includegraphics[width=0.45\linewidth]{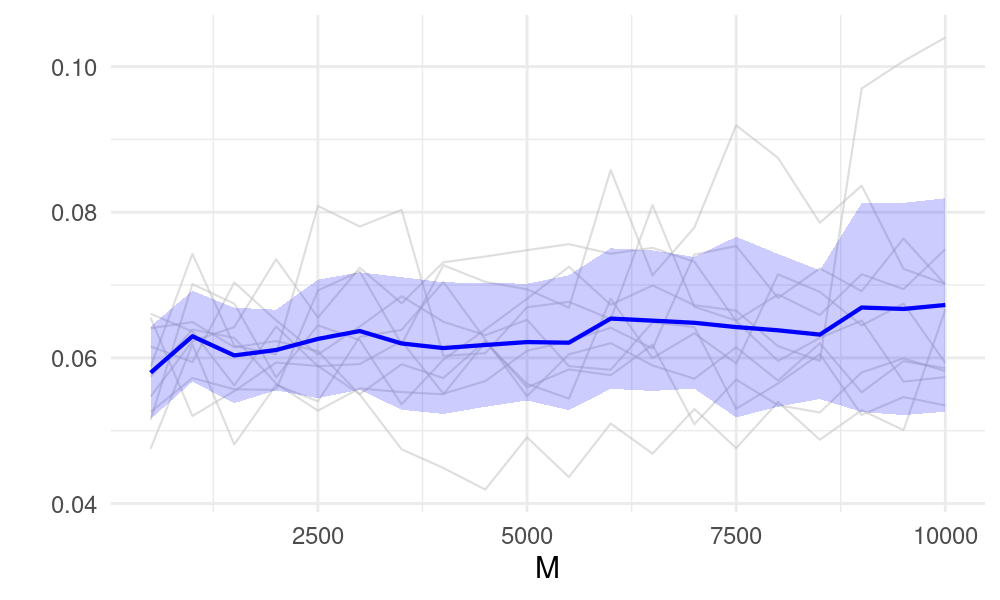} 
\subcaption{
$M_T = M^{0.2}$.
}
\label{fig: 10, 4, Case 2 MT20}
\end{subfigure}
\begin{subfigure}[b]{\textwidth}
\centering
\includegraphics[width=0.45\linewidth]{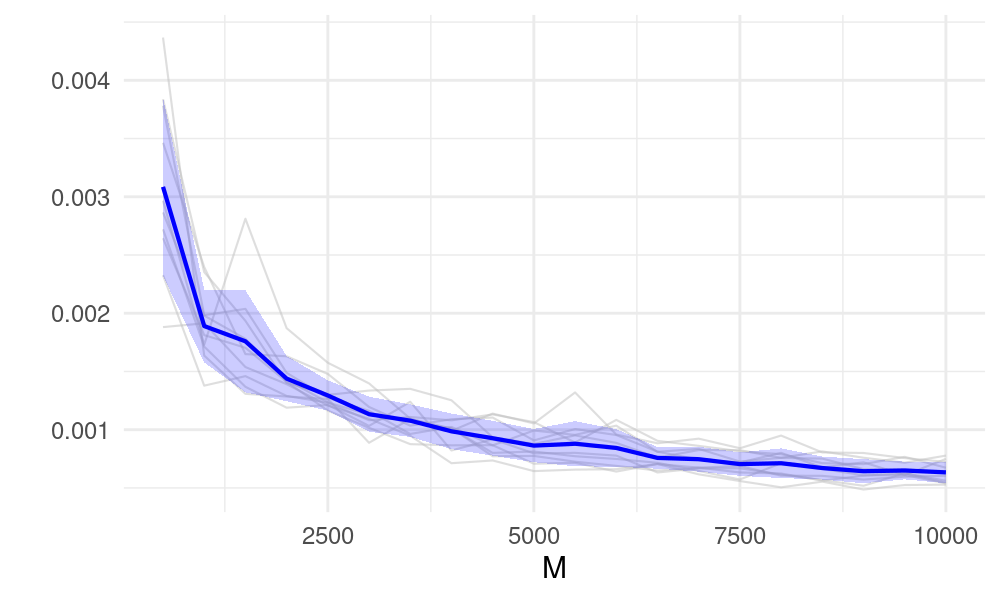} 
\includegraphics[width=0.45\linewidth]{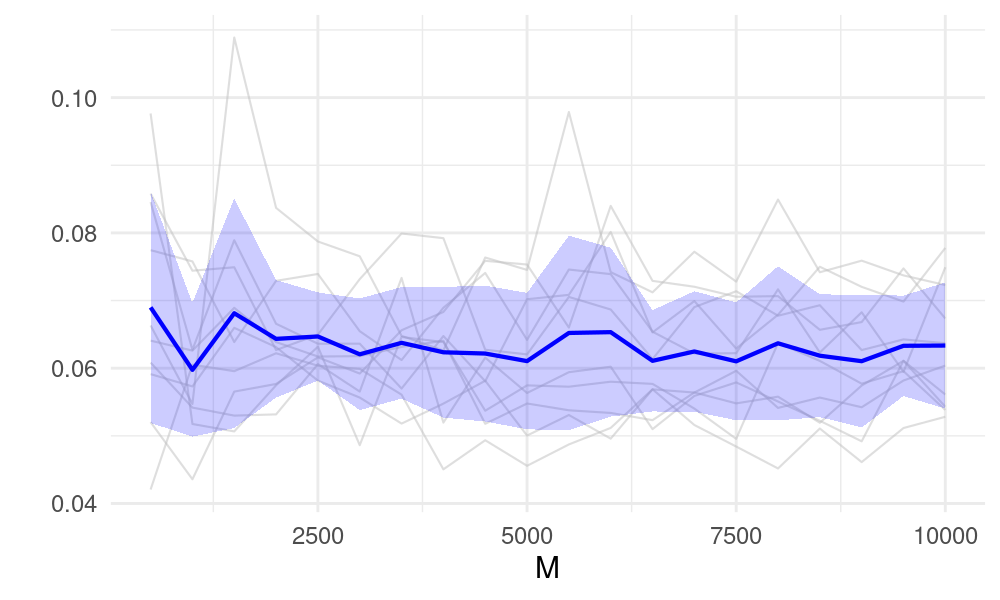}
\subcaption{
$M_T = M^{0.45}$.
}
\label{fig: 10, 4, Case 2 MT45}
\end{subfigure}
\begin{subfigure}[b]{\textwidth}
\centering
\includegraphics[width=0.45\linewidth]{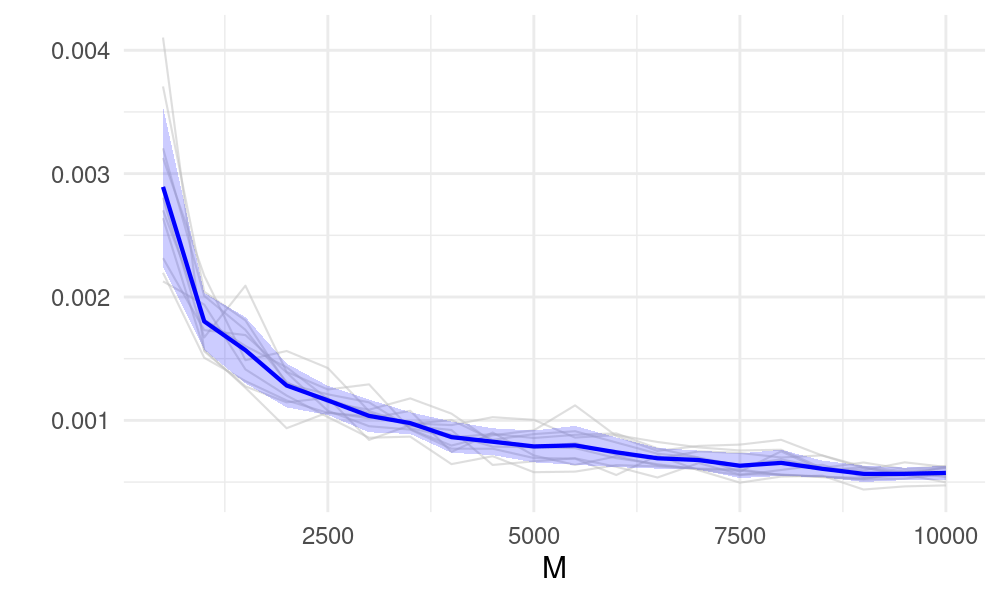} 
\includegraphics[width=0.45\linewidth]{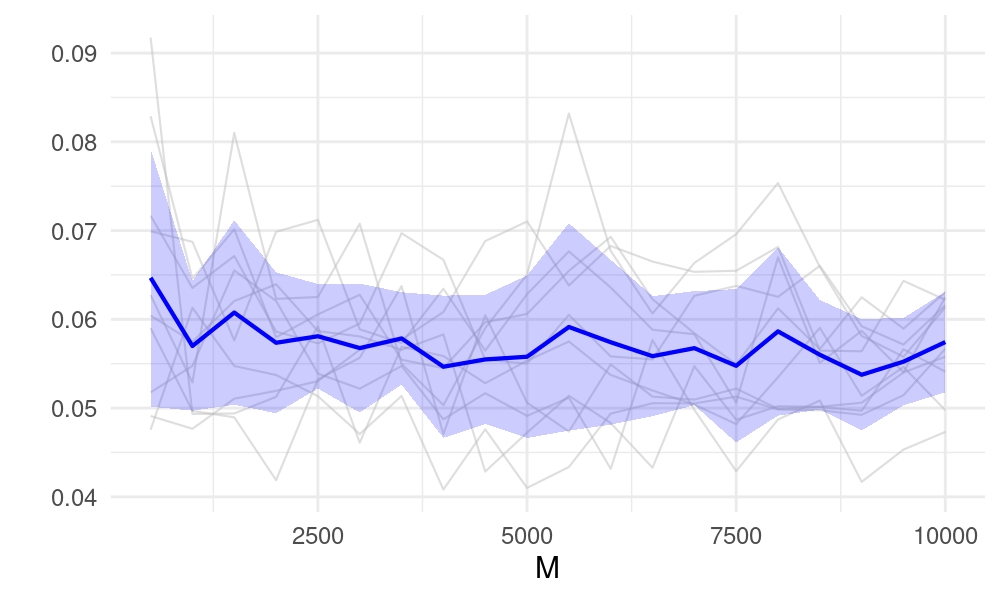} 
\subcaption{
$M_T = M^{0.8}$.
}
\label{fig: 10, 4, Case 2 MT80}
\end{subfigure}
\caption{Evolutions of $e_1^{\mathrm{SIR}}$ (left) and $\sqrt{M}e_1^{\mathrm{SIR}}$ (right) in terms of $M$ for $K=10, \nu=4, \mbox{eigen-discrepancy of Case 2}$, different clipping sizes. }
\label{fig: case 2, K=10, nu=4}
\end{figure}

The results on $K=10, \nu=4$ are reported in Figure \ref{fig: case 2, K=10, nu=4}. As previously, clipping improves the robustness of the SIR estimators, with more improvement from larger clipping sizes. In this case, $M^{0.8}$ also does not violate the convergence rate $\sqrt{M}$, and lead to smallest standard deviation. Therefore, we also suggest to use  $M^{0.8}$ to construct SIR estimation in this case. 

As for the running time of Algorithm \ref{alg: siw adapted}, we do not report in particular because it is the same as Algorithm \ref{alg: siw}, which was shown in Table \ref{tbl: runing time algo Psi general}. Note that the two algorithms share the same procedure of sampling from proposal, which determines the running time. Therefore, the proposed Algorithm \ref{alg: siw adapted} has both advantages in running time and robustness, which can always be considered used. 

\section{Conclusion}
In this work, we proposed a novel and efficient algorithm for sampling from the Shrinkage Inverse-Wishart (SIW) distribution. Unlike the existing nested Gibbs sampler, our method is based on sampling importance resampling, making it significantly faster. We provided a theoretical analysis of the convergence behavior of our algorithm which we verified across  simulations.
To further improve robustness, particularly in scenarios where importance weights exhibit large discrepancies, we introduced a clipping strategy and adapted the previous convergence results. Our empirical results validate again the theoretical findings and provide practical guidelines for selecting appropriate clipping sizes.
Overall, the proposed algorithm offers a principled and computationally attractive alternative to traditional methods for SIW sampling, with potential applications in Bayesian inference involving covariance matrix priors.

\section*{Acknowledgement}
The author acknowledges Sophie Achard and Julyan Arbel for the proofreading and fruitful discussions on the methods. The author acknowledges also Thomas Guilmeau for discussions of the convergence results of SIR estimators. 

\bibliographystyle{agsm}
\bibliography{main}

@article{berger2020Bayesian,
  author    = {Berger, James O. and Sun, Dongchu and Song, Chiara},
  title     = {Bayesian analysis of the covariance matrix of a multivariate normal distribution with a new class of priors},
  journal   = {Ann. Stat.},
  year      = {2020},
  volume    = {48},
  number    = {4},
  pages     = {2381--2403},
  doi       = {https://doi.org/10.1214/19-AOS1891}
}

@article{vazquez2017importance,
  title={Importance sampling with transformed weights},
  author={V{\'a}zquez, Manuel A and M{\'\i}guez, Joaqu{\'\i}n},
  journal={Electron. Lett.},
  volume={53},
  number={12},
  pages={783--785},
  year={2017},
  publisher={Wiley Online Library}
}

@book{cappe2005inference,
  title={Inference in hidden Markov models},
  author={Capp{\'e}, Olivier and Moulines, Eric and Ryd{\'e}n, Tobias},
  year={2005},
  publisher={Springer}
}

@book{david2004order,
  title={Order statistics},
  author={David, Herbert A and Nagaraja, Haikady N},
  year={2004},
  publisher={John Wiley \& Sons}
}

@article{tropp2012comparison,
  title={A comparison principle for functions of a uniformly random subspace},
  author={Tropp, Joel A},
  journal={Probab. Theory Relat. Fields},
  volume={153},
  number={3},
  pages={759--769},
  year={2012},
  publisher={Springer}
}

@book{chikuse2003statistics,
  title={Statistics on special manifolds},
  author={Chikuse, Yasuko},
  year={2003},
  publisher={Springer Science \& Business Media}
}

\newpage
\appendix

\section{Result on order statistics used in Algorithm \ref{alg: siw Psi=cI}}\label{thm: order stat}
\begin{theorem}[\cite{david2004order}]
Let $\bm y_i \stackrel{iid}{\sim} \pi, \, i = 1, \ldots, K,$ with some probability density function $\pi$, we have 
$$\pi(y_{(1)}, \ldots, y_{(K)}) = n!\prod_{i=1}^K\pi(y_{(i)})\bone_{\{y_{(K)} > \ldots > y_{(1)}\}}.$$
\end{theorem}

\section{Proofs in Section \ref{sec: algo Psi general}}

\begin{proof}(Proof of Proposition \ref{prop})
Consider any diagonal $\Tilde{\Delta}$ such that its $i$-th diagnoal entry denoted by $\Tilde{\lambda}_k$ satisfies $\Tilde{\lambda}_1 \geq \Tilde{\lambda}_2 \geq \ldots \geq \Tilde{\lambda}_K > 0$, and any $\Tilde{\Gamma} \in \mathcal{O}_K$. The density of  $(\Tilde{\lambda}, \Tilde{\Gamma})$ denoted by $\pi(\Tilde{\lambda}, \Tilde{\Gamma})$ then writes as 
\begin{equation}
\begin{aligned}
  \pi(\Tilde{\lambda}, \Tilde{\Gamma}) &= \sum_{I_1, \ldots, I_K}  \pi(\Tilde{\lambda}, \Tilde{\Gamma} | \Tilde{\lambda}_{i} = \lambda_{I_i}^0, \, \Tilde{\Gamma}_{:,i} = \left[\Gamma_0\right]_{:,I_i}, \, i =1, ..., K) \pi(\Tilde{\lambda}_i = \lambda_{I_i}^0, \Tilde{\Gamma}_{:,i} = \left[\Gamma_0\right]_{:,I_i}, i =1, ..., K ) \\
  &=\sum_{I_1, \ldots, I_K}  \pi(\Tilde{\lambda}_i = \lambda_{I_i}^0, \, \Tilde{\Gamma}_{:,i} = \left[\Gamma_0\right]_{:,I_i}, \, i =1, ..., K)\\
  &=\sum_{I_1, \ldots, I_K}  \pi(\Tilde{\lambda}_i = \lambda_{I_i}^0, i =1, ..., K | \,  \Tilde{\Gamma}_{:,i} = \left[\Gamma_0\right]_{:,I_i}, i =1, ..., K)\pi(  \Tilde{\Gamma}_{:,i} = \left[\Gamma_0\right]_{:,I_i}, i =1, ..., K) \\
  &=\sum_{I_1, \ldots, I_K}\prod_{i=1}^K \pi_{\mathcal{IG}}(\Tilde{\lambda}_i |\nu -1, \frac{1}{2}\Tilde{\Gamma}_{:,i}^\top \Psi \Tilde{\Gamma}_{:,i})\pi_{\mathcal{U}\{\Gamma \; \in \; \mathcal{O}_K\}} = K!\prod_{i=1}^K \pi_{\mathcal{IG}}(\Tilde{\lambda}_i |\nu -1, \frac{1}{2}\Tilde{\Gamma}_{:,i}^\top \Psi \Tilde{\Gamma}_{:,i})\pi_{\mathcal{U}\{\Gamma \; \in \; \mathcal{O}_K\}}. 
\end{aligned}
\end{equation}
Therefore, we have 
\begin{equation}
\begin{aligned}
w(\Delta,\Gamma) &\propto \frac{\prod_{i = 1}^K k_{\mathcal{IG}}(\lambda_i |\nu -1, \frac{1}{2}\Gamma_{i}^\top \Psi \Gamma_{i}) k_{\mathcal{U}\{ \mathcal{O}_K\}}(\Gamma)}{K!\prod_{i=1}^K \pi_{\mathcal{IG}}(\lambda_i |\nu -1, \frac{1}{2}\Gamma_{i}^\top \Psi \Gamma_{i})\pi_{\mathcal{U}\{\Gamma \; \in \; \mathcal{O}_K\}}} \\
&\propto \frac{\prod_{i = 1}^K k_{\mathcal{IG}}(\lambda_i |\nu -1, \frac{1}{2}\Gamma_{i}^\top \Psi \Gamma_{i}) }{\prod_{i=1}^K \pi_{\mathcal{IG}}(\lambda_i |\nu -1, \frac{1}{2}\Gamma_{i}^\top \Psi \Gamma_{i})}=  \frac{1}{\prod_{i=1}^K c_\mathcal{IG}(\nu -1, \frac{1}{2}\Gamma_{i}^\top \Psi \Gamma_{i})}.    
\end{aligned}
\end{equation} 
We recall that both $\siw(\nu,\Psi,1)$ and the uniform distribution over the Stiefel manifold $\pi_{\mathcal{U}\{\Gamma \; \in \; \mathcal{O}_K\}}$ are proper priors, in another words, have finite mass over their domains (see \cite[Theorem 1]{berger2020Bayesian} and \cite{chikuse2003statistics}). Thus the two $\propto$ above well represent a difference up to a constant. 
\end{proof}

A key fact to prove the convergence in our paper is that the proposed weight is bounded, as shown in \ref{prop: boundedness}.
\begin{prop}\label{prop: boundedness}
    $w(\Delta,\Gamma)$ is non-degenerate, furthermore it is bounded. For all samples of $(\Delta, \Gamma)$
\begin{equation}\label{eq: bound w}
  w(\Delta,\Gamma) = B\prod_{i=1}^K \frac{1}{c_\mathcal{IG}(\nu -1, \frac{1}{2}\Gamma_{i}^\top \Psi \Gamma_{i})} \leq B\frac{f_\Gamma(\nu-1)^K}{\left(\frac{1}{2} \lambda_{\min}(\Psi)\right)^{K(\nu-1)}}, 
\end{equation}
where $B$ is some constant, $f_\Gamma$ is gamma function, and $\lambda_{\min}(\Psi)$ is the smallest eigenvalue of $\Psi$. 
\end{prop}

\begin{proof}(Proof of Theorem \ref{thm: algo1 consistence})

The result can be obtained by applying Theorem 9.2.9 in \cite{cappe2005inference}. To this end, we verify the two required assumptions. 
Firstly, for Assumption 9.1.1. Note that both target distribution $\pi$ and proposal distribution $\tau$ admit probability density function, and their division gives a non-degenerate weight function $w$ in Equation \eqref{eq: weight} with $w > 0$. Thus we have $\pi \ll \tau$.
Secondly, consider $C = L_1(S^K_{++}, \tau)$, the large number law shows that the weighted sample $\{(\Tilde{\Sigma}^{(m)},1)\}_{m=1, \ldots, M}$ is consistent for $(\tau, C)$. In addition, $w$ is bounded hence $w\in C$. Thus Assumption 9.2.6 is verified.
Applying Theorem 9.2.9, we have $\{(\Sigma^{(n)}_M,1)\}_{n=1, \ldots, N}$ is consistent for $(\pi, L_1(S^K_{++}, \pi))$ as soon as $M, N$ both tend to infinity.
\end{proof}

\begin{proof}(Proof of Theorem \ref{thm: algo1 CLT})

The result can be obtained by applying Theorem 9.2.15 in \cite{cappe2005inference}. To this end, we verify the two required assumptions. 
Firstly, we have shown in the previous proof that Assumption 9.1.1. is verified. Secondly, because $w$ is bounded, it falls into $L_2(S^K_{++}, \tau)$. Thus for all $f \in \{L_2(S^K_{++}, \pi) : fw \in L_2(S^K_{++}, \tau)\}$, we have the the two CLTs in the theorem. Furthermore, in our case, $\{L_2(S^K_{++}, \pi) : fw \in L_2(S^K_{++}, \tau)\} =   L_2(S^K_{++}, \pi)$. Note that, for all $f \in L_2(S^K_{++}, \pi)$
$$\int f^2w^2 d\tau \leq B_w\int f^2 w d\tau = B_w\int f^2 \frac{d\pi}{d\tau} d\tau = B_w\int f^2 w d\tau = B_w\int f^2 d\pi < +\infty,$$
where $B_w$ is the bound on $w$ given in Equation \eqref{eq: bound w}. This completes the proof.
\end{proof}

\section{Proofs of Theorems \ref{thm: algo3 consistence} and \ref{thm: algo3 CLT}}
The proofs of the consistency and CLT of SIR estimators with clipping builds on asymptotic results of the corresponding self-normalized importance sampling estimators with clipping, which requires furthermore the results of the corresponding unnormalized importance sampling estimators with clipping. Therefore we first prove the consistency and CLT of the corresponding unnormalized importance sampling estimators with clipping in Section \ref{sec: unnormalized IS clipping}, then proceed to the results of self-normalized importance sampling estimators with clipping in Section \ref{sec: normalized IS clipping}. The final proofs of Theorems \ref{thm: algo3 consistence} and \ref{thm: algo3 CLT} are given in Section \ref{sec: main proof}. 

We start by recalling that $(\Tilde{\Sigma}^{(m)},\, w_{m}),  \; m = 1, \ldots, M$, are the intermediate outputs of the sampling procedure from Algorithm \ref{alg: siw adapted}, $w_m^\prime$ are the clipped weights defined as
\begin{equation}
    \begin{cases}
      w_{m}^\prime =  w_{(\mtl)}, \; &\mbox{if} \;  w_{m} >  w_{(\mtl)}, \\
      w_{m}^\prime = w_{m}, \; &\mbox{otherwise},
    \end{cases}
\end{equation}
and $\mtl$ be the number of clipped weights. 
\begin{definition}
Define the unnormalized importance sampling estimators with clipping 
\begin{equation}
\hat{\mu}^{UIS}_{\ml,\mtl}(f) = \frac{B}{\ml}\sum\limits_{m=1}^\ml w_m^\prime f(\Tilde{\Sigma}^{(m)}),
\end{equation}    
where $B$ is the constant given in Proposition \ref{prop}. Because the weight $w_m$ used in practice is different from the true ratio between target density and proposal density by a constant. Thus to obtain the corresponding unnormalized importance sampling estimator which needs to be consistent to $\pi(f)$, we put back the constant. However, the constant will cancel out in the normalized importance sampling estimators with clipping and the SIR estimators, thus we can ignore them when proceeding into the further proofs in Sections \ref{sec: normalized IS clipping} and \ref{sec: main proof}.  
\end{definition}
\begin{definition}
Define the normalized importance sampling estimators with clipping 
\begin{equation}
\Hat{\mu}^{IS}_{\ml,\mtl}(f) =
\frac{\sum_{m=1}^\ml w_m^\prime f(\Tilde{\Sigma}^{(m)})}{\sum_{m=1}^\ml w_m^\prime}.
\end{equation}
\end{definition}

\subsection{Asymptotic results of the unnormalized importance sampling estimators with clipping}\label{sec: unnormalized IS clipping}
Let us consider a sequence of values of $(M,\mtl,N)$, denoted by $\{(M_l,\mtl, N_l)\}_{l \in \mathbb{N^+}}$, which satisfies $\lim_{l \rightarrow \infty} M_l = \infty$, $\lim_{l \rightarrow \infty} \mtl = \infty$ and $\lim_{l \rightarrow \infty} N_l = \infty$.
\begin{lemma}\label{lem: unnormalized IS with clipping}
For any $f \in L_1(S^K_{++}, \tau)$, and $\lim_{l \rightarrow \infty} \mtl/\ml = 0$, we have 
\begin{equation}\label{eq: lem unnormalized IS clipping}
    \hat{\mu}^{UIS}_{\ml,\mtl}(f) \stackrel{P}{\longrightarrow} \pi(f).
\end{equation}
\end{lemma}
\begin{proof}
Consider the unnormalized importance sampling estimator without the clipping
\begin{equation}
      \frac{B}{\ml}\sum\limits_{m=1}^\ml w_m f(\Tilde{\Sigma}^{(m)}).
\end{equation}
Applying Theorem 9.1.2 in \cite{cappe2005inference}, we have  
\begin{equation}\label{eq: B-lem00}
      \frac{B}{\ml}\sum\limits_{m=1}^\ml w_m f(\Tilde{\Sigma}^{(m)}) \longrightarrow \pi(f), \mbox{ almost surely.}
\end{equation}
Note that to apply the theorem, two conditions are to be satisfied. Firstly, the target integral should exist. In our case, given $f$ is bounded, we have $\pi(|f|) < +\infty.$ Secondly, we have that the target probability measure is absolutely continuous with respect to the proposal measure. Note that both measures admit probability density functions, and the proposed weight is non-degenerate. 

Consider the difference between the two estimators 
\begin{equation}
\begin{aligned}
    &B\left|\frac{1}{\ml}\sum\limits_{m=1}^\ml w_m^\prime f(\Tilde{\Sigma}^{(m)}) - \frac{1}{\ml}\sum\limits_{m=1}^\ml w_m f(\Tilde{\Sigma}^{(m)})\right| \\
    =  \; &B\left| \frac{1}{\ml}\sum\limits_{w_{m} >  w_{(\mtl)}} \left(w_m^\prime f(\Tilde{\Sigma}^{(m)}) -  w_m f(\Tilde{\Sigma}^{(m)})\right)\right|   
    \leq  \frac{B}{\ml}\sum\limits_{w_{m} >  w_{(\mtl)}} \left|w_{(\mtl)} f(\Tilde{\Sigma}^{(m)}) -  w_m f(\Tilde{\Sigma}^{(m)})\right| \\
    \leq \; &\frac{B}{\ml}\sum\limits_{w_{m} >  w_{(\mtl)}} \sqrt{w_{(\mtl)}^2 +  w_m^2 } \sqrt{2}\left|   f(\Tilde{\Sigma}^{(m)})\right| \leq \frac{B}{\ml}\sum\limits_{w_{m} >  w_{(\mtl)}} 2B_w \left|   f(\Tilde{\Sigma}^{(m)})\right| ,
\end{aligned}
\end{equation}
where $B_w$ is the bound of weight defined in Equation \eqref{eq: bound w}. Apply expectation on both sides, we have
\begin{equation}
\begin{aligned}
    &\mathbb{E}B\left|\frac{1}{\ml}\sum\limits_{m=1}^\ml w_m^\prime f(\Tilde{\Sigma}^{(m)}) - \frac{1}{\ml}\sum\limits_{m=1}^\ml w_m f(\Tilde{\Sigma}^{(m)})\right| \\
    \leq \; &\frac{B}{\ml}\sum\limits_{w_{m} >  w_{(\mtl)}} 2B_w \mathbb{E}\left|   f(\Tilde{\Sigma}^{(m)})\right| \leq B\frac{\mtl}{\ml} 2B_w\mathbb{E}\left|   f(\Tilde{\Sigma}^{(1)})\right|.
\end{aligned}
\end{equation}
Note that $\left|   f(\Tilde{\Sigma}^{(m)})\right|, m = 1, ..., \ml $ are iid, and $\mathbb{E}\left|   f(\Tilde{\Sigma}^{(1)})\right| < +\infty$ due to the $L_1$ assumption.
Therefore, given $\mtl/\ml \rightarrow 0$, we have 
\begin{equation}\label{eq: B-lem1}
    \mathbb{E}B\left|\frac{1}{\ml}\sum\limits_{m=1}^\ml w_m^\prime f(\Tilde{\Sigma}^{(m)}) - \frac{1}{\ml}\sum\limits_{m=1}^\ml w_m f(\Tilde{\Sigma}^{(m)})\right| \longrightarrow 0.
\end{equation}
Combine Equations \eqref{eq: B-lem00} and \eqref{eq: B-lem1}, we obtain the convergence in probability. 
\end{proof}

\begin{lemma}\label{lem: unnormalized IS with clipping CLT}
For any $f \in L_1(S^K_{++}, \tau)$, and $\lim_{l \rightarrow \infty} \mtl/\sqrt{\ml} = 0$, we have 
\begin{equation}
    \sqrt{\ml}\left(\hat{\mu}^{UIS}_{\ml,\mtl}(f) -  \pi(f)\right) \stackrel{\mathcal{D}}{\longrightarrow} \mathcal{N}(0, \mbox{var}_{\tau}(f w)),
\end{equation}
where $\mbox{var}_{\tau}(f w) = \mathbb{V} [w_mf(\Tilde{\Sigma}^{(m)})]$. The variance is well-defined because both weight and $f$ are bounded, and all $\Tilde{\Sigma}^{(m)}$ are iid.
\end{lemma}
\begin{proof}
    Consider the unnormalized importance sampling estimator without the clipping
\begin{equation}
      \frac{B}{\ml}\sum\limits_{m=1}^\ml w_m f(\Tilde{\Sigma}^{(m)}).
\end{equation}
Applying Theorem 9.1.2 in \cite{cappe2005inference}, we have  
\begin{equation}\label{eq: B-lem0}
      \sqrt{\ml}\left(\frac{B}{\ml}\sum\limits_{m=1}^\ml w_m f(\Tilde{\Sigma}^{(m)}) - \pi(f) \right)\stackrel{\mathcal{D}}{\longrightarrow} \mathcal{N}(0, \mbox{var}_{\tau}(f w)).
\end{equation}
The boundedness of $f$ and the non-degenerate weight validate the conditions of the theorem. Note that
\begin{equation}
\begin{aligned}
    \hat{\mu}^{UIS}_{\ml,\mtl}(f)- \pi(f) = &\frac{B}{\ml}\sum\limits_{m=1}^\ml w_m f(\Tilde{\Sigma}^{(m)}) - \pi(f) + B\Delta_l, 
\end{aligned}
\end{equation}
where $\Delta_l = \left(\frac{1}{\ml}\sum\limits_{m=1}^\ml w_m^\prime f(\Tilde{\Sigma}^{(m)}) - \frac{1}{\ml}\sum\limits_{m=1}^\ml w_m f(\Tilde{\Sigma}^{(m)})\right).$
Multiply by $\sqrt{\ml}$ on both sides, we have 
\begin{equation}\label{eq: idk name}
\begin{aligned}
    \sqrt{\ml}\left(\hat{\mu}^{UIS}_{\ml,\mtl}(f)- \pi(f)\right) = &\sqrt{\ml}\left(\frac{B}{\ml}\sum\limits_{m=1}^\ml w_m f(\Tilde{\Sigma}^{(m)}) - \pi(f)\right) + \sqrt{\ml}B\Delta_l.
\end{aligned}
\end{equation}
From Lemma \ref{lem: unnormalized IS with clipping}, there is 
\begin{equation}
\begin{aligned}
    &\sqrt{\ml}\mathbb{E}\left|\Delta_l\right| \leq \frac{\mtl}{\sqrt{\ml}} 2B_w\mathbb{E}\left|   f(\Tilde{\Sigma}^{(1)})\right| \longrightarrow 0, \; \mbox{ when }\frac{\mtl}{\sqrt{\ml}} \rightarrow 0.
\end{aligned}
\end{equation}
Thus $\sqrt{\ml}B\Delta_l \stackrel{P}{\longrightarrow} 0.$ Combine the convergence with Equations \eqref{eq: B-lem0} and \eqref{eq: idk name}, the desired result is obtained.
\end{proof}


\subsection{Asymptotic results of the normalized importance sampling estimators with clipping}\label{sec: normalized IS clipping}
\begin{lemma}\label{lem: normalized IS with clipping}
For any $f \in L_1(S^K_{++}, \tau)$, and $\lim_{l\rightarrow \infty}\mtl/\ml = 0$, we have 
\begin{equation}
     \Hat{\mu}^{IS}_{\ml,\mtl}(f) \stackrel{P}{\longrightarrow}  \pi(f).
\end{equation}
\end{lemma}
\begin{proof}
Applying Lemma \ref{lem: unnormalized IS with clipping} with $f \equiv 1$, we have
\begin{equation}\label{eq: lem normalized IS clipping}
    \frac{B}{\ml}\sum\limits_{m=1}^\ml w_m^\prime \stackrel{P}{\longrightarrow} 1.
\end{equation}
Noticing that 
\begin{equation}
     \frac{\sum_{m=1}^\ml w_m^\prime f(\Tilde{\Sigma}^{(m)})}{\sum_{m=1}^\ml w_m^\prime} = \frac{\frac{B}{\ml}\sum_{m=1}^\ml w_m^\prime f(\Tilde{\Sigma}^{(m)})}{\frac{B}{\ml}\sum_{m=1}^\ml w_m^\prime}.
\end{equation}
Therefore, by using Slutsky's theorem on Equations \eqref{eq: lem normalized IS clipping} and \eqref{eq: lem unnormalized IS clipping}, we have the desired result.
\end{proof}
\begin{lemma}\label{lem: normalized IS with clipping CLT}
For any $f \in L_2(S^K_{++}, \tau)$, and $ \lim_{l \rightarrow \infty}\mtl/\sqrt{\ml} = 0$, we have 
\begin{equation}
    \sqrt{\ml}\left(\Hat{\mu}^{IS}_{\ml,\mtl}(f) -  \pi(f)\right) \stackrel{\mathcal{D}}{\longrightarrow} \mathcal{N}(0, \mbox{var}_{\tau}(f w - \pi(f)w)),
\end{equation}
where $\mbox{var}_{\tau}(f w - \pi(f)) = \mathbb{V} [w_1f(\Tilde{\Sigma}^{(1)}) - \pi(f)w_1]^2$. The variance is well-defined because $f \in L_2(S^K_{++}, \tau)$ and $w$ is bounded, in addition, all $\Tilde{\Sigma}^{(m)}$ are iid.
\end{lemma}
\begin{proof}
    Note that 
\begin{equation}
     \sqrt{\ml}\left(\frac{\sum_{m=1}^\ml w_m^\prime f(\Tilde{\Sigma}^{(m)})}{\sum_{m=1}^\ml w_m^\prime} - \pi(f)\right) = \frac{\frac{B}{\sqrt{\ml}}\sum_{m=1}^\ml w_m^\prime \left(f(\Tilde{\Sigma}^{(m)})-\pi(f)\right)}{\frac{B}{\ml}\sum_{m=1}^\ml w_m^\prime}.
\end{equation}
For the denominator, applying Lemma \ref{lem: unnormalized IS with clipping} with $f \equiv 1$, we have
    $\frac{B}{\ml}\sum\limits_{m=1}^\ml w_m^\prime  \stackrel{P}{\longrightarrow} 1$.
Note that for probability measure, we have $L_2 \subset L_1$. 
For the numerator, applying Lemma \ref{lem: unnormalized IS with clipping CLT} with $\Tilde{f} = f - \pi(f)$, we have, given $\mtl/\sqrt{\ml} \rightarrow 0$,
\begin{equation}
    \sqrt{\ml}\left(\frac{B}{\ml}\sum_{m=1}^\ml w_m^\prime \Tilde{f}(\Tilde{\Sigma}^{(m)})-\pi(\Tilde{f})\right) \stackrel{\mathcal{D}}{\longrightarrow} \mathcal{N}(0,\mbox{var}_{\tau}(\Tilde{f}w)).
\end{equation}
Note that $$\pi(\Tilde{f}) = \int \Tilde{f} d \pi = \int f - \pi(f) d \pi  = \int f d \pi - \pi(f) =0.$$ 
We have 
\begin{equation}
    \frac{B}{\sqrt{\ml}}\sum_{m=1}^\ml w_m^\prime \left(f(\Tilde{\Sigma}^{(m)}) - \pi(f)\right) \stackrel{\mathcal{D}}{\longrightarrow} \mathcal{N}(0,\mbox{var}_{\tau}(\Tilde{f}w)).
\end{equation}
Therefore, by using Slutsky's theorem, we have the desired CLT.
\end{proof}

\subsection{Asymptotic results of the sampling importance resampling estimators with clipping}\label{sec: main proof}
In this section, we provide proofs of the main theorems for Algorithm \ref{alg: siw adapted}.
\begin{proof}(Proof of Theorem \ref{thm: algo3 consistence})

We will apply Theorem 9.2.8 in \cite{cappe2005inference} by taking $(X, \mathcal{X})$, $\mu$, $\mathcal{F}_l$, and $\{\xi^{l,n}\}_{1\leq n \leq N_l}$ respectively as $\pi$, $(S^K_{++},\mathcal{B})$, $\sigma\{(\Tilde{\Sigma}^{(1)}, \ldots, \Tilde{\Sigma}^{(M_l)})\}$ and $\{\Sigma_{M_l, M_{T,l}}^{(n)}\}_{1\leq n \leq N_l}$. To this end, we need to verify its assumptions. Firstly, by the principle of resampling we have, given $\mathcal{F}_l$, $\{\Sigma_{M_l,M_{T,l}}^{(n)}\}_{1\leq n \leq N_l}$ are conditionally independent. Secondly, we need to show for any $f \in L_1(S^K_{++}, \tau)$ and $C \geq 0$, it holds 
\begin{equation}\label{eq: a messy eq}
    \frac{1}{N_l} \sum\limits_{n=1}^{N_l} \mathbb{E} \left[|f|\left(\Sigma_{M_l,M_{T,l}}^{(n)}\right)\bone_{\left\{|f|\left(\Sigma_{M_l,M_{T,l}}^{(n)}\right) \geq C\right\}} \bigg| \mathcal{F}_l\right] \stackrel{P}{\longrightarrow} \pi(|f| \bone_{\left\{|f| \geq C\right\}}).
\end{equation}
We first show when without clipping the resampling samples from the same proposal samples $(\Tilde{\Sigma}^{(1)}, \ldots, \Tilde{\Sigma}^{(M_l)})$ satisfy Equation \eqref{eq: a messy eq}. The proof of Theorem \ref{thm: algo1 consistence} shows that Assumption 9.1.1 and Assumption 9.2.6 are verified with $C = L_1(S^K_{++},\tau)$ for our proposal samples $\Tilde{\Sigma}^{(m)}, m = 1, \ldots, M_l$. Thus applying Theorem 9.2.7, we have for any $f \in L_1(S^K_{++},\ms)$
\begin{equation}\label{eq: a messy eq1}
\begin{aligned}
    &\frac{1}{N_l} \sum\limits_{n=1}^{N_l} \mathbb{E} \left[|f|\left(\Sigma_{M_l}^{(n)}\right)\bone_{\left\{|f|\left(\Sigma_{M_l}^{(n)}\right) \geq C\right\}} \bigg| \mathcal{F}_l\right] =
     \mathbb{E} \left[|f|\left(\Sigma_{M_l}^{(n)}\right)\bone_{\left\{|f|\left(\Sigma_{M_l}^{(n)}\right) \geq C\right\}} \bigg| \mathcal{F}_l\right] \\
    & = \frac{\sum_{m=1}^{M_l} w_m |f|\left(\Tilde{\Sigma}^{(m)}\right)\bone_{\left\{|f|\left(\Tilde{\Sigma}^{(m)}\right) \geq C\right\}}}{\sum_{m=1}^\ml w_m}  \stackrel{P}{\longrightarrow} \pi(|f| \bone_{\left\{|f| \geq C\right\}}).    
\end{aligned}
\end{equation}
On the other hand, for any $f \in L_1(S^K_{++}, \tau)$, 
$$\int |f| d\pi = \int |f| \frac{d\pi}{d\tau} d\tau =\int |f| w d\tau \leq B_w\int |f| d\tau < +\infty,$$
where $B_w$ is the bound on $w$ given in Equation \eqref{eq: bound w}. Thus in our case, $L_1(S^K_{++}, \tau) \subseteq L_1(S^K_{++},\ms)$, therefore Equation \eqref{eq: a messy eq1} holds for any $f \in L_1(S^K_{++}, \tau)$. Now consider the difference between the two left hand sides of Equations \eqref{eq: a messy eq} and \eqref{eq: a messy eq1} for any $f \in L_1(S^K_{++}, \tau)$:
\begin{equation}
\begin{aligned}
   &\left|\frac{1}{N_l} \sum\limits_{n=1}^{N_l} \mathbb{E} \left[|f|\left(\Sigma_{M_l,M_{T,l}}^{(n)}\right)\bone_{\left\{|f|\left(\Sigma_{M_l,M_{T,l}}^{(n)}\right) \geq C\right\}} \bigg| \mathcal{F}_l\right] -  \frac{1}{N_l} \sum\limits_{n=1}^{N_l} \mathbb{E} \left[|f|\left(\Sigma_{M_l}^{(n)}\right)\bone_{\left\{|f|\left(\Sigma_{M_l}^{(n)}\right) \geq C\right\}} \bigg| \mathcal{F}_l\right] \right|\\
   &= \left|\mathbb{E} \left[|f|\left(\Sigma_{M_l,M_{T,l}}^{(n)}\right)\bone_{\left\{|f|\left(\Sigma_{M_l,M_{T,l}}^{(n)}\right) \geq C\right\}} \bigg| \mathcal{F}_l\right] -  \mathbb{E} \left[|f|\left(\Sigma_{M_l}^{(n)}\right)\bone_{\left\{|f|\left(\Sigma_{M_l}^{(n)}\right) \geq C\right\}} \bigg| \mathcal{F}_l\right] \right|\\
   &= \left|\frac{\sum_{m=1}^\ml w_m^\prime |f|\left(\Tilde{\Sigma}^{(m)}\right)\bone_{\left\{|f|\left(\Tilde{\Sigma}^{(m)}\right) \geq C\right\}}}{\sum_{m=1}^\ml w_m^\prime} - \frac{\sum_{m=1}^{M_l} w_m |f|\left(\Tilde{\Sigma}^{(m)}\right)\bone_{\left\{|f|\left(\Tilde{\Sigma}^{(m)}\right) \geq C\right\}}}{\sum_{m=1}^\ml w_m} \right|.
\end{aligned}
\end{equation}
We bound the difference above with $I_l^1 + I_l^2$, where 
\begin{equation}
I_l^1  = \left|\frac{\sum_{m=1}^\ml w_m^\prime |f|\left(\Tilde{\Sigma}^{(m)}\right)\bone_{\left\{|f|\left(\Tilde{\Sigma}^{(m)}\right) \geq C\right\}}}{\sum_{m=1}^\ml w_m^\prime} - \frac{\sum_{m=1}^{M_l} w_m |f|\left(\Tilde{\Sigma}^{(m)}\right)\bone_{\left\{|f|\left(\Tilde{\Sigma}^{(m)}\right) \geq C\right\}}}{\sum_{m=1}^\ml w_m^\prime} \right|,
\end{equation}
and 
\begin{equation}
I_l^2  = \left|\frac{\sum_{m=1}^\ml w_m|f|\left(\Tilde{\Sigma}^{(m)}\right)\bone_{\left\{|f|\left(\Tilde{\Sigma}^{(m)}\right) \geq C\right\}}}{\sum_{m=1}^\ml w_m^\prime} - \frac{\sum_{m=1}^{M_l} w_m |f|\left(\Tilde{\Sigma}^{(m)}\right)\bone_{\left\{|f|\left(\Tilde{\Sigma}^{(m)}\right) \geq C\right\}}}{\sum_{m=1}^\ml w_m} \right|.
\end{equation}
Firstly, 
\begin{equation}
\begin{aligned}
I_l^1  &= \left|\frac{\sum_{m=1}^\ml w_m^\prime |f|\left(\Tilde{\Sigma}^{(m)}\right)\bone_{\left\{|f|\left(\Tilde{\Sigma}^{(m)}\right) \geq C\right\}} - \sum_{m=1}^{M_l} w_m |f|\left(\Tilde{\Sigma}^{(m)}\right)\bone_{\left\{|f|\left(\Tilde{\Sigma}^{(m)}\right) \geq C\right\}}}{\sum_{m=1}^\ml w_m^\prime} \right|\\
& = \left|\frac{\sum_{w_{m} >  w_{(\mtl)}} (w_m^\prime - w_m) |f|\left(\Tilde{\Sigma}^{(m)}\right)\bone_{\left\{|f|\left(\Tilde{\Sigma}^{(m)}\right) \geq C\right\}}}{\sum_{m=1}^\ml w_m^\prime} \right| \\
& \leq \frac{\sum_{w_{m} >  w_{(\mtl)}} |w_m^\prime - w_m |\,|f|\left(\Tilde{\Sigma}^{(m)}\right)\bone_{\left\{|f|\left(\Tilde{\Sigma}^{(m)}\right) \geq C\right\}}}{\sum_{m=1}^\ml w_m^\prime} \\&\leq 2\frac{B_w}{B} \frac{\sum_{w_{m} >  w_{(\mtl)}} |f|\left(\Tilde{\Sigma}^{(m)}\right)}{\sum_{m=1}^\ml w_m^\prime} = 2\frac{B_w}{B} \frac{\sum_{w_{m} >  w_{(\mtl)}} |f|\left(\Tilde{\Sigma}^{(m)}\right)}{\ml\frac{1}{\ml}\sum_{m=1}^\ml w_m^\prime}.
\end{aligned}
\end{equation}
On one hand, given $\lim_{l \rightarrow \infty} \frac{\mtl}{\ml} =0$, $\mtl \mathbb{E} |f|\left(\Tilde{\Sigma}^{(m)}\right)/\ml \longrightarrow 0.$ Therefore, 
$$\frac{\sum_{w_{m} >  w_{(\mtl)}} |f|\left(\Tilde{\Sigma}^{(m)}\right)}{\ml} \stackrel{P}{\longrightarrow} 0.$$
On the other hand, from Equation \eqref{eq: lem normalized IS clipping}, we have 
$ \frac{B}{\ml}\sum_{m=1}^\ml w_m^\prime \stackrel{P}{\longrightarrow} 1.$
Therefore, given $\lim_{l \rightarrow \infty} \frac{\mtl}{\ml} =0$, $I^1_l \stackrel{P}{\longrightarrow} 0$. 

Secondly,
\begin{equation}
\begin{aligned}
I_l^2  &= \left|\left(\frac{1}{\sum_{m=1}^\ml w_m^\prime} - \frac{1}{\sum_{m=1}^\ml w_m} \right) \sum_{m=1}^\ml w_m|f|\left(\Tilde{\Sigma}^{(m)}\right)\bone_{\left\{|f|\left(\Tilde{\Sigma}^{(m)}\right) \geq C\right\}} \right|\\
&=\left|\frac{\sum_{w_{m} >  w_{(\mtl)}} (w_{m} - w_{m}^\prime)}{\sum_{m=1}^\ml w_m^\prime \sum_{m=1}^\ml w_m} \sum_{m=1}^\ml w_m|f|\left(\Tilde{\Sigma}^{(m)}\right)\bone_{\left\{|f|\left(\Tilde{\Sigma}^{(m)}\right) \geq C\right\}} \right|\\
&= \left|\frac{\sum_{w_{m} >  w_{(\mtl)}} (w_{m} - w_{m}^\prime)}{\ml \frac{1}{\ml}\sum_{m=1}^\ml w_m^\prime } \frac{\sum_{m=1}^\ml w_m|f|\left(\Tilde{\Sigma}^{(m)}\right)\bone_{\left\{|f|\left(\Tilde{\Sigma}^{(m)}\right) \geq C\right\}}}{\sum_{m=1}^\ml w_m} \right|\\
&\leq \frac{\mtl 2B_w}{\ml \frac{B}{\ml}\sum_{m=1}^\ml w_m^\prime }  \frac{\sum_{m=1}^\ml w_m|f|\left(\Tilde{\Sigma}^{(m)}\right)\bone_{\left\{|f|\left(\Tilde{\Sigma}^{(m)}\right) \geq C\right\}}}{\sum_{m=1}^\ml w_m} 
\end{aligned}
\end{equation}
Given $\lim_{l \rightarrow \infty} \frac{\mtl}{\ml} =0$, the upper bound above converges to $0$ in probability. Therefore,  $ I^2_l \stackrel{P}{\rightarrow} 0$. Thus for any $f \in L_1(S^K_{++}, \tau)$:
\begin{equation}\label{eq: convg diff}
    \left|\frac{1}{N_l} \sum\limits_{n=1}^{N_l} \mathbb{E} \left[|f|\left(\Sigma_{M_l,M_{T,l}}^{(n)}\right)\bone_{\left\{|f|\left(\Sigma_{M_l,M_{T,l}}^{(n)}\right) \geq C\right\}} \bigg| \mathcal{F}_l\right] -  \frac{1}{N_l} \sum\limits_{n=1}^{N_l} \mathbb{E} \left[|f|\left(\Sigma_{M_l}^{(n)}\right)\bone_{\left\{|f|\left(\Sigma_{M_l}^{(n)}\right) \geq C\right\}} \bigg| \mathcal{F}_l\right] \right| \stackrel{P}{\rightarrow} 0.
\end{equation}
Combine Equations \eqref{eq: a messy eq1} and \eqref{eq: convg diff}, we obtain Equation \eqref{eq: a messy eq}. Applying Theorem 9.2.8, we have 
\begin{equation}
\begin{aligned}
    &\frac{1}{\nl}\sum\limits_{n=1}^\nl \left\{f\bigl(\Sigma^{(n)}_{\ml,\mtl}\bigr) - \mathbb{E} \left[f\left(\Sigma_{M_l,M_{T,l}}^{(n)}\right) \bigg| \mathcal{F}_l\right]\right\} = \frac{1}{\nl}\sum\limits_{n=1}^\nl f\bigl(\Sigma^{(n)}_{\ml,\mtl}\bigr) - \mathbb{E} \left[f\left(\Sigma_{M_l,M_{T,l}}^{(n)}\right) \bigg| \mathcal{F}_l\right]\\
     & = \frac{1}{\nl}\sum\limits_{n=1}^\nl f\bigl(\Sigma^{(n)}_{\ml,\mtl}\bigr) - \frac{\sum_{m=1}^{M_l} w_m^\prime  f\left(\Tilde{\Sigma}^{(m)}\right)}{\sum_{m=1}^\ml w_m^\prime}  \stackrel{P}{\longrightarrow} 0.    
\end{aligned}
\end{equation}
Lastly from Lemma \ref{lem: normalized IS with clipping}, we have 
$$\frac{\sum_{m=1}^{M_l} w_m^\prime  f\left(\Tilde{\Sigma}^{(m)}\right)}{\sum_{m=1}^\ml w_m^\prime}  \stackrel{P}{\longrightarrow} \pi(f).$$
This completes the proof.
\end{proof}

\begin{proof}(Proof of Theorem \ref{thm: algo3 CLT})

Write $\Hat{\mu}^\mathrm{SIR}_{M_l,M_{T_l},N_l}(f) -  \pi(f)$ as $A_l + B_l$, where 
\begin{equation}
    A_l = \frac{\sum_{m=1}^{M_l} w_m^\prime  f\left(\Tilde{\Sigma}^{(m)}\right)}{\sum_{m=1}^\ml w_m^\prime} - \pi(f),
\end{equation}
and 
\begin{equation}
 B_l = \frac{1}{\nl}\sum\limits_{n=1}^\nl \left\{f\bigl(\Sigma^{(n)}_{\ml,\mtl}\bigr) - \mathbb{E}\left[f\bigl(\Sigma^{(n)}_{\ml,\mtl}\bigr) \big| \mathcal{F}_l \right]\right\}, 
\end{equation}
with $\mathcal{F}_l = \sigma\{(\Tilde{\Sigma}^{(1)}, \ldots, \Tilde{\Sigma}^{(M_l)})\}$. By Lemma \ref{lem: normalized IS with clipping CLT}, we have for any $f \in L_2(S^K_{++}, \tau)$, when $ \mtl/\sqrt{\ml} \rightarrow 0$
\begin{equation}\label{eq: Al CLT}
    \sqrt{\ml}A_l \stackrel{\mathcal{D}}{\longrightarrow} \mathcal{N}(0, \mbox{var}_{\tau}(f w - \pi(f)w)). 
\end{equation}
To derive a CLT for $B_l$, we will apply Theorem 9.5.13 in \cite{cappe2005inference} by taking $(X, \mathcal{F})$, $\mu$, $\mathcal{F}_l$, and $\{\xi^{l,n}\}_{1\leq n \leq N_l}$ respectively as $\pi$, $(S^K_{++},\mathcal{B})$, $\sigma\{(\Tilde{\Sigma}^{(1)}, \ldots, \Tilde{\Sigma}^{(M_l)})\}$ and $\{\Sigma_{M_l, M_{T,l}}^{(n)}\}_{1\leq n \leq N_l}$. To this end, we need to check the three conditions with $f \in L_2(S^K_{++}, \tau)$. Firstly, $\{\Sigma_{M_l, M_{T,l}}^{(n)}\}_{1\leq n \leq N_l}$ are conditionally independent given $\mathcal{F}_l$ for any $l$. In addition, for any $n = 1, \ldots, \nl$, $\mathbb{E}\left[f^2(\Sigma_{M_l, M_{T,l}}^{(n)})\big | \mathcal{F}_l\right] = \Hat{\mu}^{IS}_{\ml,\mtl}(f^2) < +\infty$. Secondly, 
\begin{equation}
\begin{aligned}
    &\frac{1}{N_l} \sum\limits_{n=1}^{N_l} \left\{\mathbb{E} \left[f^2\left(\Sigma_{M_l,M_{T,l}}^{(n)}\right) \bigg| \mathcal{F}_l\right] - \mathbb{E}^2 \left[f\left(\Sigma_{M_l,M_{T,l}}^{(n)}\right) \bigg| \mathcal{F}_l\right]\right\}\\
    & = \left\{\mathbb{E} \left[f^2\left(\Sigma_{M_l,M_{T,l}}^{(1)}\right) \bigg| \mathcal{F}_l\right] - \mathbb{E}^2 \left[f\left(\Sigma_{M_l,M_{T,l}}^{(1)}\right) \bigg| \mathcal{F}_l\right]\right\} \\
    &= \frac{\sum_{m=1}^{M_l} w_m^\prime  f^2\left(\Tilde{\Sigma}^{(m)}\right)}{\sum_{m=1}^\ml w_m^\prime} - \left[\frac{\sum_{m=1}^{M_l} w_m^\prime  f\left(\Tilde{\Sigma}^{(m)}\right)}{\sum_{m=1}^\ml w_m^\prime}\right]^2 \stackrel{P}{\longrightarrow} \mbox{var}_\pi(f),
\end{aligned}
\end{equation}
where $\mbox{var}_\pi(f) = \pi(f^2) - \left[\pi(f) \right]^2$. The variance is well defined thanks to the $L_2$ assumption. In fact,
for any $f \in L_2(S^K_{++}, \tau)$, 
$$\int f^2 d\pi = \int f^2 \frac{d\pi}{d\tau} d\tau =\int f^2 w d\tau \leq B_w\int f^2 d\tau < +\infty,$$
where $B_w$ is the bound on $w$ given in Equation \eqref{eq: bound w}. 
The last equation is obtained by applying Lemma \ref{lem: normalized IS with clipping}. Thirdly, take $\mu$ of condition $(iii)$ as $\pi$, thus $f \in L_2(S^K_{++}, \tau) \subset L_2(S^K_{++}, \pi)$. The proof of Theorem \ref{thm: algo3 consistence} shows that 
\begin{equation}
\frac{1}{N_l} \sum\limits_{n=1}^{N_l} \mathbb{E} \left[f^2\left(\Sigma_{M_l}^{(n)}\right)\bone_{\left\{|f|\left(\Sigma_{M_l}^{(n)}\right) \geq C\right\}} \bigg| \mathcal{F}_l\right] \stackrel{P}{\longrightarrow} \pi(f^2\bone_{\left\{|f| \geq C\right\}}).    
\end{equation}
By applying Theorem 8.5.13 in \cite{cappe2005inference}, we have for any real $u$
\begin{equation}\label{eq: Bl CLT}\mathbb{E}\left[\exp{\left(iu\sqrt{\nl}B_l\right)}\bigg| \mathcal{F}_l\right] \stackrel{P}{\longrightarrow} \exp(-\mbox{var}_\pi(f)u^2/2).
\end{equation}
Therefore, for any reals $u$ and $v$,
\begin{equation}   
\begin{aligned}
&\mathbb{E}\left[\exp{\left\{i\left(u\sqrt{\nl}B_l + v\sqrt{\ml}A_l\right)\right\}}\right]  \\
=&\mathbb{E}\left[\mathbb{E}\left[\exp{\left\{iu\sqrt{\nl}B_l  \right\}}\bigg| \mathcal{F}_l \right]\exp{\left\{iv\sqrt{\ml}A_l\right\}}\right] \\
& \longrightarrow \exp(-\mbox{var}_\pi(f)u^2/2)\exp(-\mbox{var}_{\tau}(f w - \pi(f)w)v^2/2).
\end{aligned}
\end{equation}
The point-wise convergence comes from Equations \eqref{eq: Al CLT} and \eqref{eq: Bl CLT} as well as the fact that $|\exp{\left\{iu\sqrt{\nl}B_l  \right\}}| = 1$. 
Thus the bivariate characteristic function converges to the characteristic function of a bivariate normal, implying that
\begin{equation}
\begin{pmatrix}
  \sqrt{\nl}B_l \\ \sqrt{\ml}A_l 
\end{pmatrix}
\stackrel{\mathcal{D}}{\longrightarrow} \mathcal{N}\left(0, 
  \begin{bmatrix}
\mbox{var}_\pi(f) & 0 \\ 0 & \mbox{var}_{\tau}(f w - \pi(f)w)
  \end{bmatrix}
  \right)
\end{equation}
Put $b_l$ = $\sqrt{\ml}$ if $\alpha < 1$ and $b_l = \sqrt{\nl}$ if $\alpha \geq 1$. The proof follows from $b_l (A_l + B_l) = (b_l \ml^{-1/2}) \sqrt{\ml}A_l + (b_l \nl^{-1/2})\sqrt{\nl}B_l$.
    
\end{proof}

\section{Proof of Theorem \ref{thm: existence E}}\label{sec: proof of existence E}
\begin{proof}
Note that for any $p \geq 1$
\begin{equation}
\E\left[\bm\Sigma^p\right] = \E\left[\bm\Gamma \bm\Delta^p\bm\Gamma^\top\right] = \E\left[\sum_{k=1}^K\bm\lambda_k^p\bm\Gamma_{:,k} \bm\Gamma_{:,k}^\top\right].
\end{equation}
Therefore 
\begin{equation}
\E\left|(\bm\Sigma_{ij})^p\right| = \E\left|\sum_{k=1}^K\bm\lambda_k^p\bm\Gamma_{i,k} \bm\Gamma_{j,k}\right|  \leq   \E\sum_{k=1}^K\bm\lambda_k^p\left|\bm\Gamma_{i,k} \bm\Gamma_{j,k}\right| \leq   \E\sum_{k=1}^K\bm\lambda_k^p.
\end{equation}
The last inequality is because $\bm\Gamma$ is orthonormal, thus its entities are less than $1$ in absolute value. For each $\E\bm\lambda_k^p$, $\E\bm\lambda_k^p = \E\left[\E\bm\lambda_k^p | \bm\Gamma\right]$. When $p < \nu - 1$, for a given $\bm\Gamma$, $\E\,\bm\lambda_k^p | \bm\Gamma < +\infty$. Thus 
\begin{equation}
    \E\left|(\bm\Sigma_{ij})^p\right| \leq   \sum_{k=1}^K\E\bm\lambda_k^p = \sum_{k=1}^K\E\left[\E\,\bm\lambda_k^p | \bm\Gamma \right] < +\infty.
\end{equation}
Similarly, $\E\left|(\bm\Sigma^{-p})_{ij}\right| \leq   \E\sum_{k=1}^K\bm\lambda_k^{-p}$. When $\nu > 1$, $ \E\left[\E\bm\lambda_k^{-p} | \bm\Gamma\right] < +\infty$ for a given $\bm\Gamma$, which implies $\E\left|(\bm\Sigma^{-p})_{ij}\right| < +\infty$. 
\end{proof}

\section{Evolutions of $e_1^{\mathrm{SIR}}$ for $K=100$, Algorithm \ref{alg: siw adapted}}
The results related to $K=100$ of $e_1^{\mathrm{SIR}}$ from Section \ref{sec: eva algo adapted} are given in Figures \ref{fig: case 2, K=100, nu=4} and \ref{fig: case 2, K=100, nu=20}.
\begin{figure}[htbp]
\centering 
\begin{subfigure}[b]{\textwidth}
\centering
\includegraphics[width=0.45\linewidth]{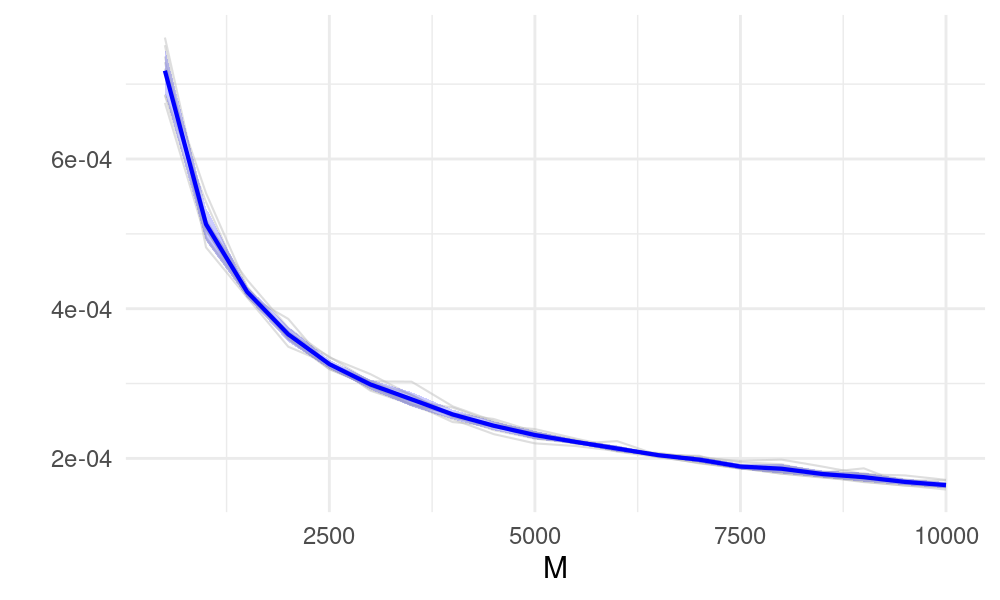} 
\includegraphics[width=0.45\linewidth]{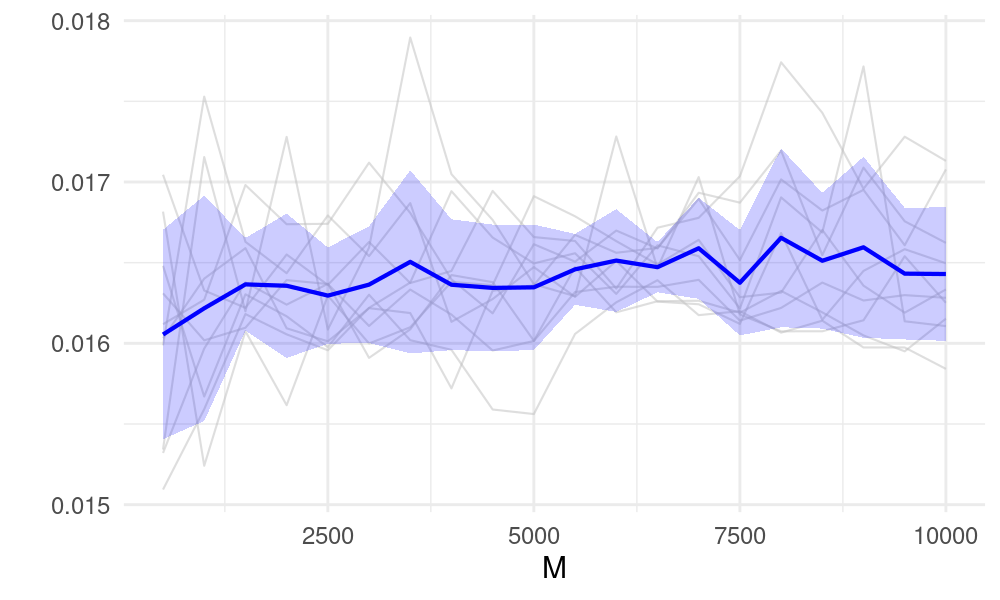} 
\subcaption{
$M_T = M^{0.2}$.
}
\label{fig: 100, 4, Case 2 MT20}
\end{subfigure}
\begin{subfigure}[b]{\textwidth}
\centering
\includegraphics[width=0.45\linewidth]{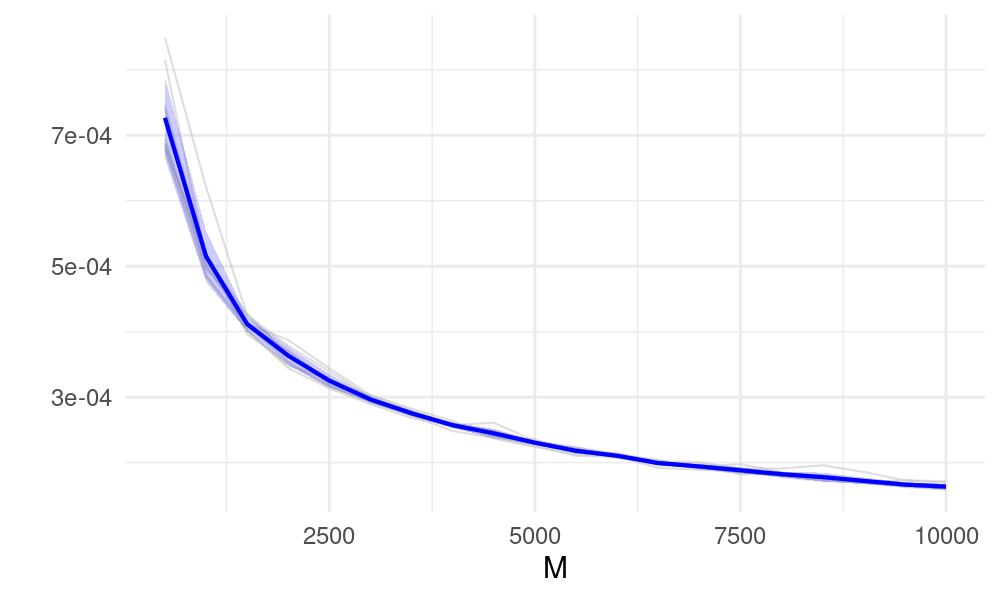} 
\includegraphics[width=0.45\linewidth]{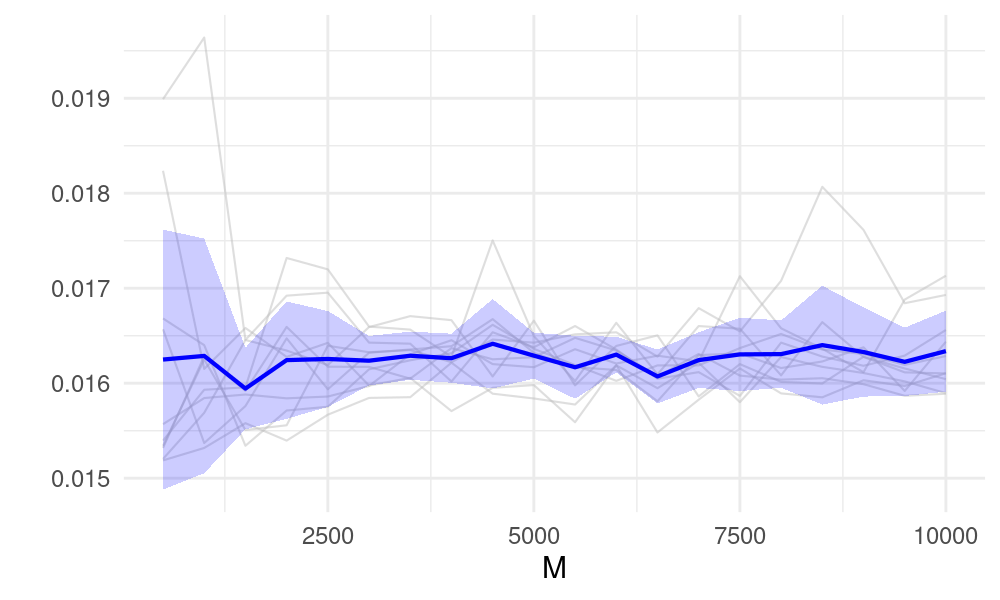}
\subcaption{
$M_T = M^{0.45}$.
}
\label{fig: 100, 4, Case 2 MT45}
\end{subfigure}
\begin{subfigure}[b]{\textwidth}
\centering
\includegraphics[width=0.45\linewidth]{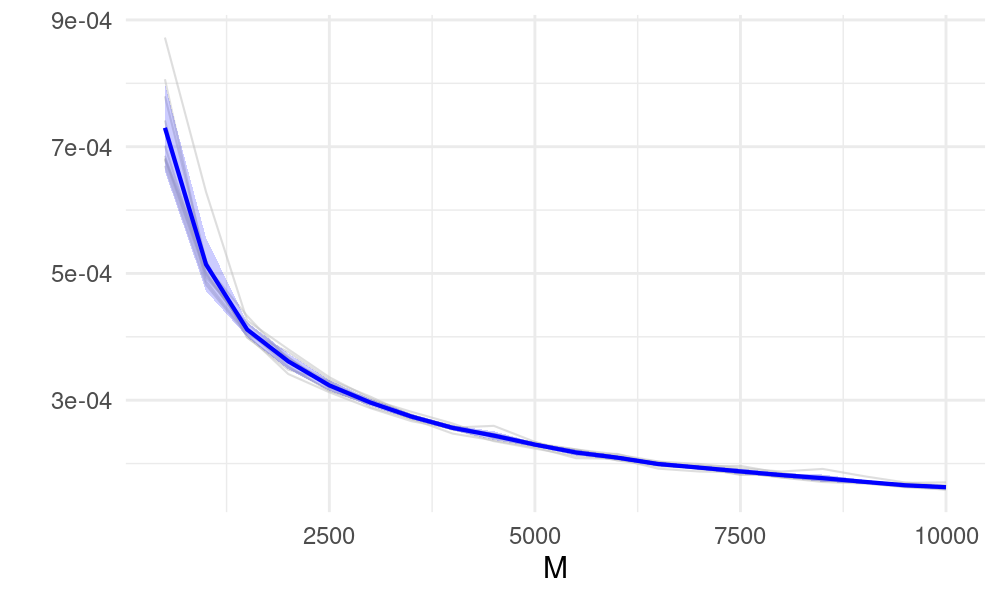} 
\includegraphics[width=0.45\linewidth]{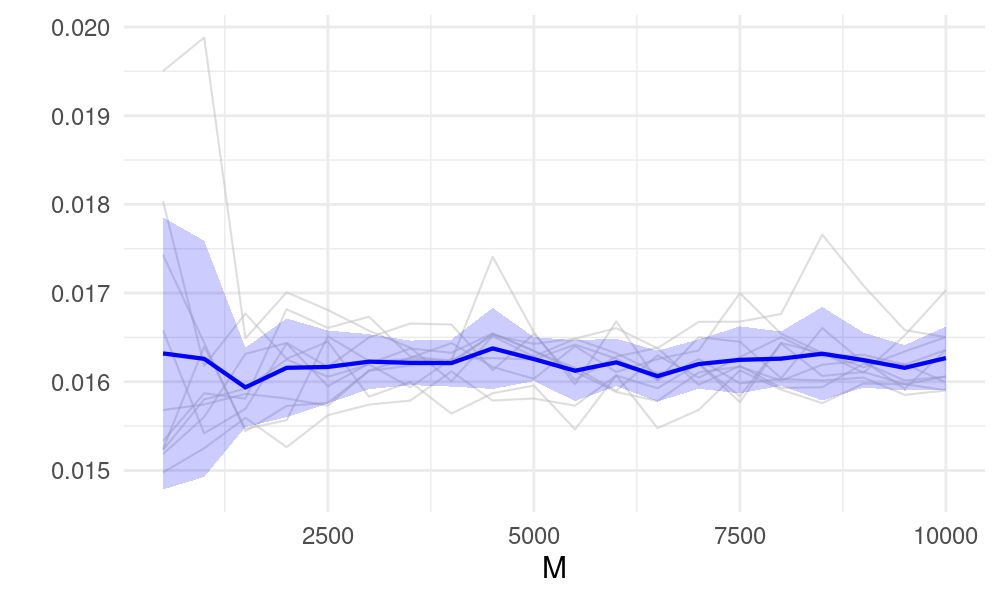} 
\subcaption{
$M_T = M^{0.8}$.
}
\label{fig: 100, 4, Case 2 MT80}
\end{subfigure}
\caption{Evolutions of $e_1^{\mathrm{SIR}}$ (left) and $\sqrt{M}e_1^{\mathrm{SIR}}$ (right) in terms of $M$ for $K=100, \nu=20, \mbox{eigen-discrepancy of Case 2}$, different clipping sizes. }
\label{fig: case 2, K=100, nu=4}
\end{figure}
\begin{figure}[htbp]
\centering 
\begin{subfigure}[b]{\textwidth}
\centering
\includegraphics[width=0.45\linewidth]{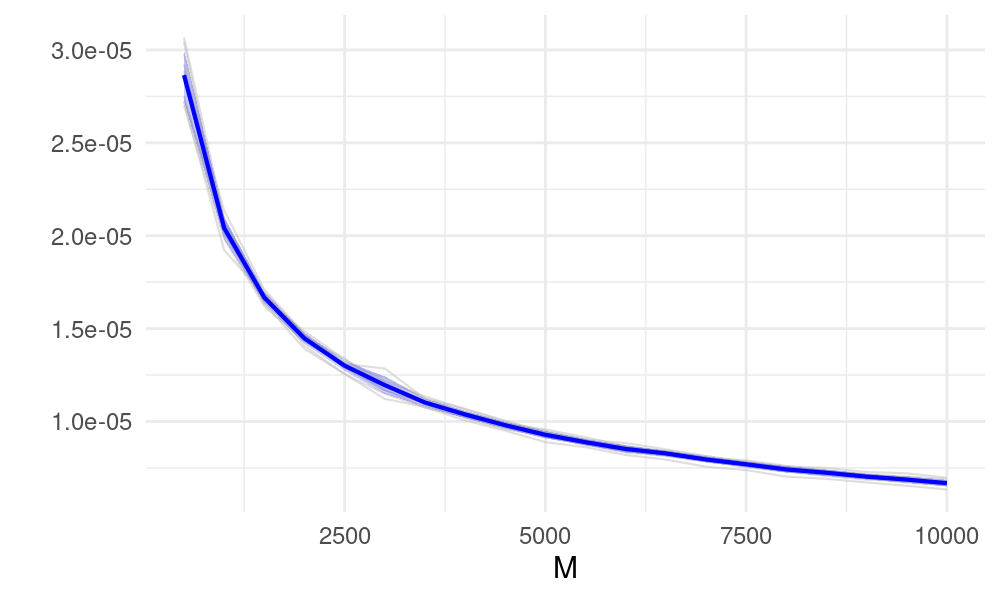} 
\includegraphics[width=0.45\linewidth]{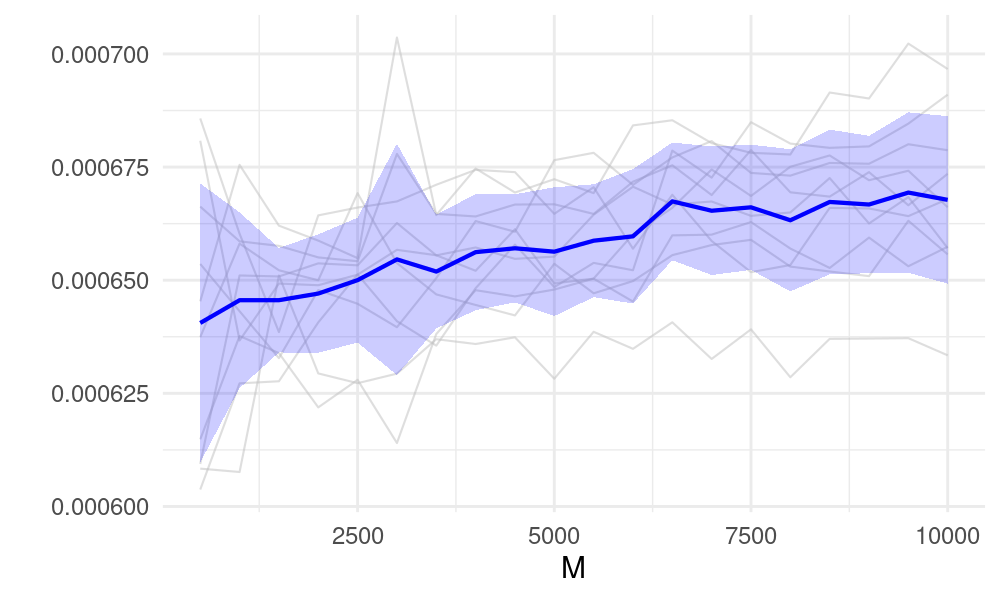} 
\subcaption{
$M_T = M^{0.2}$.
}
\label{fig: 100, 20, Case 2 MT20}
\end{subfigure}
\begin{subfigure}[b]{\textwidth}
\centering
\includegraphics[width=0.45\linewidth]{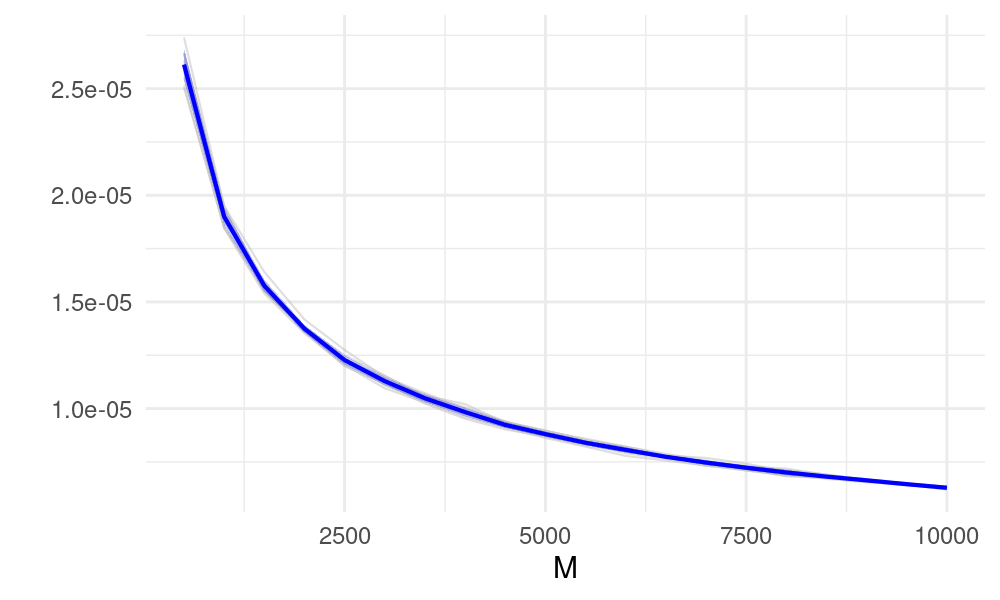} 
\includegraphics[width=0.45\linewidth]{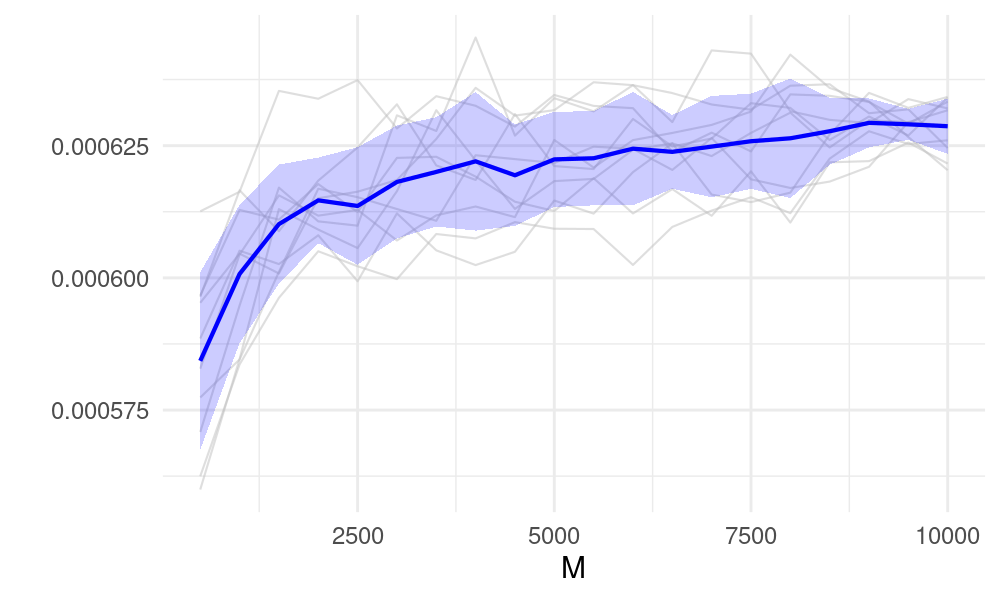}
\subcaption{
$M_T = M^{0.45}$.
}
\label{fig: 100, 20, Case 2 MT45}
\end{subfigure}
\begin{subfigure}[b]{\textwidth}
\centering
\includegraphics[width=0.45\linewidth]{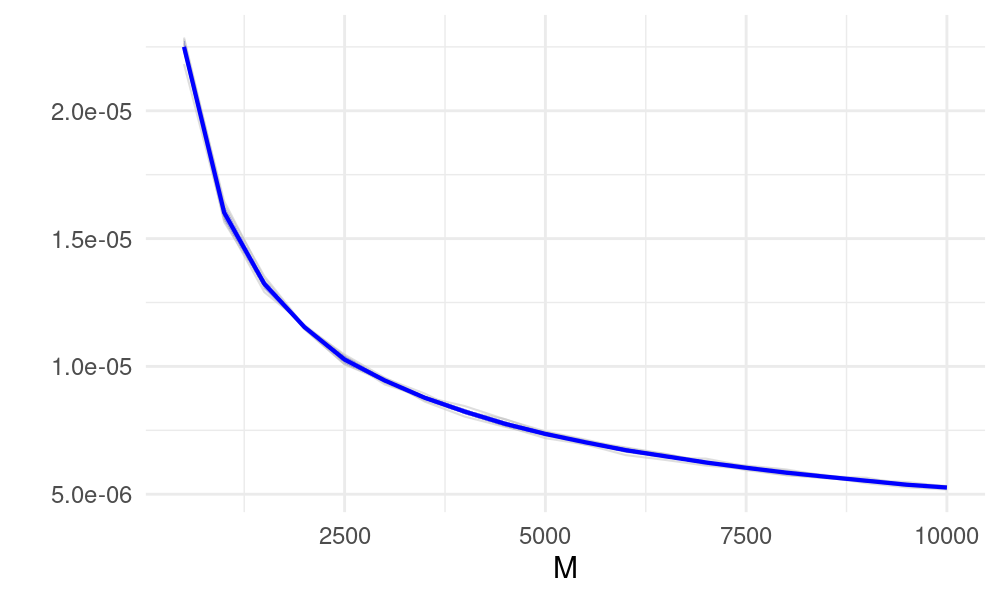} 
\includegraphics[width=0.45\linewidth]{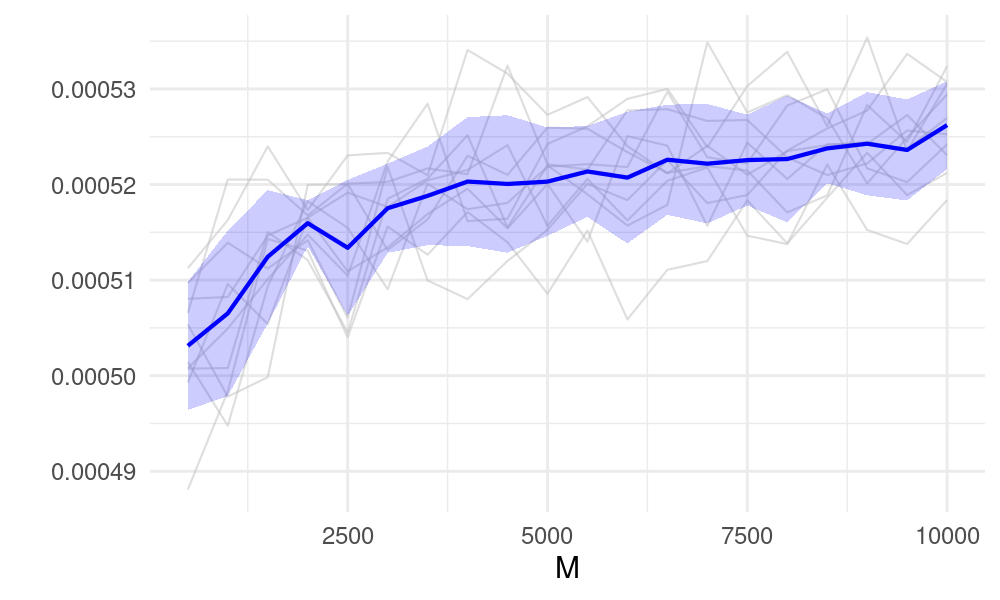} 
\subcaption{
$M_T = M^{0.8}$.
}
\label{fig: 100, 20, Case 2 MT80}
\end{subfigure}
\caption{Evolutions of $e_1^{\mathrm{SIR}}$ (left) and $\sqrt{M}e_1^{\mathrm{SIR}}$ (right) in terms of $M$ for $K=100, \nu=20, \mbox{eigen-discrepancy of Case 2}$, different clipping sizes. }
\label{fig: case 2, K=100, nu=20}
\end{figure}

\end{document}